\documentclass[twocolumn,10pt]{IEEEtran}

\usepackage{cite}
\usepackage[acronym]{glossaries}
\usepackage{amssymb}
\usepackage{amsmath}
\usepackage{amsthm}
\usepackage{graphicx}
\usepackage{setspace}
\usepackage[linesnumbered,ruled]{algorithm2e}
\usepackage{xcolor}
\usepackage{url}

\newacronym{awgn}{AWGN}{additive white Gaussian noise}
\newacronym{csi}{CSI}{channel state information}
\newacronym{los}{LoS}{line-of-sight}
\newacronym{nlos}{NLoS}{non-line-of-sight}
\newacronym{mimo}{MIMO}{multiple-input multiple-output}
\newacronym{miso}{MISO}{multiple-input single-output}
\newacronym{rf}{RF}{radio frequency}
\newacronym{siso}{SISO}{single-input single-output}
\newacronym{mrt}{MRT}{maximum ratio transmission}
\newacronym{wpt}{WPT}{wireless power transfer}
\newacronym{star-ris}{STAR-RIS}{simultaneously transmitting and reflecting RIS}
\newacronym{ios}{IOS}{intelligent omni-surface}
\newacronym{swipt}{SWIPT}{simultaneous wireless information and power transfer}
\newacronym{hmimos}{HMIMOS}{holographic multiple-input multiple-output surface}
\newacronym{rsma}{RSMA}{rate splitting multiple access}

\newtheorem{proposition}{Proposition}

\begin{document}

\title{Closed-Form Global Optimization of Beyond Diagonal Reconfigurable Intelligent Surfaces}

\author{Matteo Nerini,~\IEEEmembership{Graduate Student Member,~IEEE},
        Shanpu Shen,~\IEEEmembership{Senior Member,~IEEE},\\
        Bruno Clerckx,~\IEEEmembership{Fellow,~IEEE}

\thanks{This work was supported by Hong Kong Research Grants Council through the Collaborative Research Fund under Grant C6012-20G. (\textit{Corresponding author: Shanpu Shen}.)}
\thanks{M. Nerini is with the Department of Electrical and Electronic Engineering, Imperial College London, London, SW7 2AZ, U.K. (e-mail: m.nerini20@imperial.ac.uk).}
\thanks{S. Shen is with the Department of Electronic and Computer Engineering, The Hong Kong University of Science and Technology, Clear Water Bay, Kowloon, Hong Kong (e-mail: sshenaa@connect.ust.hk).}
\thanks{B. Clerckx is with the Department of Electrical and Electronic Engineering, Imperial College London, London SW7 2AZ, U.K. and with Silicon Austria Labs (SAL), Graz A-8010, Austria (e-mail: b.clerckx@imperial.ac.uk).}}

\maketitle

\begin{abstract}
Reconfigurable intelligent surfaces (RISs) allow controlling the propagation environment in wireless networks by tuning multiple reflecting elements.
RISs have been traditionally realized through single connected architectures, mathematically characterized by a diagonal scattering matrix.
Recently, beyond diagonal RISs (BD-RISs) have been proposed as a novel branch of RISs whose scattering matrix is not limited to be diagonal, which creates new benefits and opportunities for RISs.
Efficient BD-RIS architectures have been realized based on group and fully connected reconfigurable impedance networks.
However, a closed-form solution for the global optimal scattering matrix of these architectures is not yet available.
In this paper, we provide such a closed-form solution proving that the theoretical performance upper bounds can be exactly achieved for any channel realization.
We first consider the received signal power maximization in single-user \gls{siso} systems aided by a BD-RIS working in reflective or transmissive mode.
Then, we extend our solution to single-user \gls{mimo} and multi-user \gls{miso} systems.
We show that our algorithm is less complex than the iterative optimization algorithms employed in the previous literature.
The complexity of our algorithm grows linearly (resp. cubically) with the number of RIS elements in the case of group (resp. fully) connected architectures.
\end{abstract}

\glsresetall


\section{Introduction}

Reconfigurable intelligent surfaces (RISs) are an emerging technology that will enhance the performance of future wireless communications \cite{wu19a,bas19,hua19,wu21,dir20}.
This technology relies on large planar surfaces comprising multiple reflecting elements, each of them capable of inducing a certain amplitude and phase change to the incident electromagnetic wave.
Thus, an RIS can steer the reflected signal toward the intended direction by smartly coordinating the reflection coefficients of its elements.
RIS-aided communication systems benefit from several advantages.
RISs with passive elements are characterized by ultra-low power consumption and do not cause any active additive thermal noise or self-interference phenomena.
Furthermore, RIS is a low-profile and cost-effective solution since it does not include expensive \gls{rf} chains.
In conventional RISs, denoted as single connected RISs, each element is controlled by a tunable impedance connected to ground \cite{she20}.
As a result, conventional RISs are characterized by a diagonal scattering matrix, also known as phase shift matrix.

Conventional RISs have been optimized with several objectives, such as transmit power minimization \cite{wu19b}, weighted sum-power minimization \cite{liu21}, and weighted sum-rate maximization \cite{guo20}.
In \cite{li21}, RISs have been designed to optimally support wide-band communications.
Recently, RISs have been also applied to improve the efficiency of \gls{wpt} \cite{fen22} and \gls{swipt} systems \cite{zha22a}.
Multi-RIS aided systems have been studied in \cite{li20,zhe21,mei22}, where the inter-RIS signal reflections are exploited to fully unveil the potential of this technology.
Path-loss models for RISs considering both near-field and far-field propagation have been developed in \cite{ozd20,tan21}.
Since continuous phase shifts are hard to realize in practice, RISs have been designed based on discrete phase shifts \cite{wu19c,di20}.
In \cite{zhe20,wei21,guo22}, the authors addressed the problem of low-overhead channel estimation in RIS-aided systems.
In \cite{abe20,cai20}, practical reflection models capturing the phase-dependent amplitude variation in the reflection coefficients have been developed.
Finally, prototypes of discrete phase shift RISs have been designed in \cite{dai20,dun20}.

Differently from conventional RISs, beyond diagonal RISs (BD-RISs) have been proposed as a novel branch of RISs in which the scattering matrix is not limited to be diagonal \cite{li23-1}.
Several BD-RIS architectures have been introduced, as shown in the classification tree in Fig.~\ref{fig:ris-tree}.
In \cite{she20}, the authors generalized the single connected architecture by connecting all or a subset of RIS elements through a reconfigurable impedance network, resulting in the fully and group connected architecture, respectively.
Group and fully connected RISs have been designed with discrete reflection coefficients in \cite{ner21}.
In \cite{xu21}, the concept of \gls{star-ris}, or \gls{ios}, has been introduced.
This BD-RIS architecture is able to reflect and transmit the impinging signal, differently from conventional RISs working only in reflective mode \cite{xu22,zha20,zha22b}.
In \cite{li22-1}, a general RIS model has been proposed to unify different modes (reflective/transmissive/hybrid) and different architectures (single/group/fully connected).
The authors also propose the novel cell-wise group/fully connected BD-RIS architecture.
In \cite{li22-2}, multi-sector BD-RISs have been proposed to achieve full-space coverage.
The synergy between multi-sector BD-RISs and \gls{rsma} proved to improve the performance, coverage, and save on antennas in multi-user systems \cite{li23-2}.
In \cite{li22-3}, dynamically group connected RISs are optimized based on a dynamic grouping strategy.
In \cite{li22}, a BD-RIS architecture with a non-diagonal phase shift matrix is proposed, able to achieve a higher rate than conventional RISs.
Several benefits of BD-RISs over conventional RISs can be identified.
Since BD-RISs can adjust not only the phases but also the magnitudes of the impinging waves, the received signal power is consequently improved \cite{she20}.
In group connected RISs, the grouping strategy can be properly optimized to further increase the received signal power \cite{ner21,li22-3}.
When discrete reflection coefficients are considered, BD-RISs achieve the performance upper bound with fewer resolution bits than conventional RISs \cite{ner21}.
Finally, BD-RISs enable efficient hybrid transmissive and reflective mode \cite{li22-1}, and highly directional full-space coverage \cite{li22-2}.

\begin{figure}[t]
\centering
\includegraphics[width=0.48\textwidth]{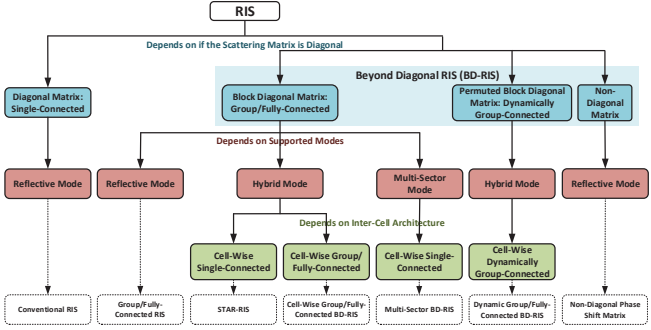}
\caption{RIS classification tree.}
\label{fig:ris-tree}
\end{figure}

The fully connected architecture enables the best performance gain with respect to all other RIS models proposed to date \cite{she20}.
This is due to the additional degrees of freedom provided by the complex architecture.
Besides, the group connected architecture has been proposed to achieve a good trade-off between performance enhancement and complexity.
Depending on the group size, this architecture bridges between the single and the fully connected ones.
However, a closed-form solution for the global optimal scattering matrix of group and fully connected architectures is not yet available.
The scattering matrix has been optimized in recent literature by employing costly iterative optimization algorithms \cite{she20,ner21}.
For this reason, it was possible to show that the theoretical performance upper bounds are tight only numerically.

In this paper, we provide a closed-form global optimal solution for the scattering matrix of group and fully connected RISs, driven by the success of these novel BD-RIS architectures.
The resulting scattering matrix is proved to exactly achieve the received signal power upper bounds derived in \cite{she20} for \gls{siso} systems.
Thus, we mathematically prove that these upper bounds are tight.
Furthermore, we show that our algorithm is less complex than the iterative optimization methods applied to design the scattering matrix in the recent literature \cite{she20}.
The complexity of our algorithm grows linearly with the number of RIS elements in the case of group connected architectures, while it grows cubically in fully connected architectures.
Given the non-convexity of the involved optimization problems, this is the first study deriving a closed-form global optimal solution for BD-RISs.
Our solution for single-user \gls{siso} systems is also proven to be general enough to allow the optimization of multiple problems in RIS-aided multi-antenna systems.
The contributions of this paper are summarized as follows:

\textit{First}, as the main contribution, we provide a low-complexity closed-form global optimal solution for the scattering matrix of BD-RISs working in reflective mode applied to single-user \gls{siso} systems.
In these systems, group and fully connected RISs designed with our solution exactly achieve their performance upper bounds.
The upper bound-achieving property of our solution is valid for any channel realization, with no assumptions on its distribution and correlation.

\textit{Second}, we consider BD-RISs working in transmissive mode, enabled by the cell-wise group connected architecture proposed in \cite{li22-1}.
We show how our optimal solution can be exploited to also globally optimize these BD-RISs.
Also in the case of transmissive mode, the performance upper bounds are always exactly achieved by BD-RISs optimized through our solution.

\textit{Third}, we exploit our optimal solution to optimize the RIS scattering matrix in single-user \gls{mimo} systems, including \gls{miso} systems as a special case.
For systems aided by a fully connected RIS and with negligible direct link, we derive a tight upper bound on the received signal power.
We show that such an upper bound can be always exactly achieved with our optimal strategy.
In addition, we also propose an efficient sub-optimal solution for the case in which the direct link is not negligible.
In this case, a tight upper bound on the received signal power is not known.

\textit{Fourth}, we study the weighted sum power maximization problem in multi-user \gls{miso} systems.
In the case of systems aided by a fully connected RIS and with negligible direct links, we provide a tight performance upper bound and an optimal solution to achieve it.
Also for multi-user \gls{miso} systems, we provide a sub-optimal solution to design the RIS in the case the direct links are not negligible.
In fact, a tight performance upper bound is not available in this case.

\textit{Organization}: In Section~\ref{sec:system-model}, we define the system model and the problem formulation.
In Section~\ref{sec:design}, we derive the upper bound-achieving closed-form solution for the scattering matrix in single-user \gls{siso} systems.
In Section~\ref{sec:design-transmissive}, we show that our solution can be also applied to optimally design BD-RISs working in transmissive mode.
In Sections~\ref{sec:su-mimo} and \ref{sec:mu-mimo}, we extend our closed-form solution to single-user \gls{mimo}, and multi-user \gls{miso} systems, respectively.
In Section~\ref{sec:results}, we assess the obtained performance through numerical simulations.
Finally, Section~\ref{sec:conclusion} contains the concluding remarks.
For reproducible research, the simulation code is available at \url{https://github.com/matteonerini/optimization-of-bdris}.

\textit{Notation}: Vectors and matrices are denoted with bold lower and bold upper letters, respectively.
Scalars are represented with letters not in bold font.
$\left|a\right|$, and $\arg\left(a\right)$ refer to the modulus and phase of a complex scalar $a$, respectively.
$\left[\mathrm{\mathbf{a}}\right]_{i}$ and $\left\Vert\mathbf{a}\right\Vert$ refer to the $i$th element and $l_{2}$-norm of vector $\mathrm{\mathbf{a}}$, respectively.
$\mathbf{A}^{*}$, $\mathbf{A}^{T}$, $\mathbf{A}^{H}$, $\left[\mathrm{\mathbf{A}}\right]_{i,j}$, and $\left\Vert\mathbf{A}\right\Vert$ refer to the conjugate, transpose, conjugate transpose, $\left(i,j\right)$th element, and $l_{2}$-norm of a matrix $\mathbf{A}$, respectively.
$\mathbf{u}_{\text{max}}(\mathbf{A})$ and $\mathbf{v}_{\text{max}}(\mathbf{A})$ denote the dominant left and right singular vectors of a matrix $\mathbf{A}$, respectively.
$\mathbb{R}$ and $\mathbb{C}$ denote the real and complex number sets, respectively.
$j=\sqrt{-1}$ denotes imaginary unit.
$\mathbf{0}$ and $\mathbf{I}$ denote an all-zero matrix and an identity matrix, respectively, with appropriate dimensions.
$\mathcal{CN}\left(\mathbf{0},\mathbf{I}\right)$ denotes the distribution of a circularly symmetric complex Gaussian random vector with mean vector $\mathbf{0}$ and covariance matrix $\mathbf{I}$ and $\sim$ stands for \textquotedblleft distributed as\textquotedblright.
diag$\left(a_{1},\ldots,a_{N}\right)$ refers to a diagonal matrix with diagonal elements being $a_{1},\ldots,a_{N}$.
diag$\left(\mathrm{\mathbf{A}}_{1},\ldots,\mathrm{\mathbf{A}}_{N}\right)$ refers to a block diagonal matrix with blocks being $\mathrm{\mathbf{A}}_{1},\ldots,\mathrm{\mathbf{A}}_{N}$.

\section{BD-RIS-Aided System Model}
\label{sec:system-model}

Let us consider a single-user \gls{siso} scenario in which the communication is aided by an $N_{I}$ antenna RIS.
The $N_{I}$ antennas of the RIS are connected to a $N_{I}$-port reconfigurable impedance network, with scattering matrix $\boldsymbol{\Theta}\in\mathbb{C}^{N_{I}\times N_{I}}$.
As widely adopted in the related literature, we assume no mutual coupling between the RIS antennas.
Defining $x\in\mathbb{C}$ as the transmitted signal and $y\in\mathbb{C}$ as the received signal, we have $y=hx+n$, where $n$ is the \gls{awgn} at the receiver.
The channel $h$ can be written as
\begin{equation}
h=h_{RT}+\mathbf{h}_{RI}\boldsymbol{\Theta}\mathbf{h}_{IT},\label{eq:H}
\end{equation}
where $h_{RT}\in\mathbb{C}$, $\mathbf{h}_{RI}\in\mathbb{C}^{1\times N_{I}}$, and $\mathbf{h}_{IT}\in\mathbb{C}^{N_{I}\times 1}$ refer to the channels from the transmitter to receiver, from the RIS to the receiver, and from the transmitter to the RIS, respectively.
According to network theory \cite{poz11}, denoting with $\mathbf{Z}_{I}\in\mathbb{C}^{N_{I}\times N_{I}}$ the impedance matrix of the $N_{I}$-port reconfigurable impedance network, $\boldsymbol{\Theta}$ can be expressed as $\boldsymbol{\Theta}=\left(\mathbf{Z}_{I}+Z_{0}\mathbf{I}\right)^{-1}\left(\mathbf{Z}_{I}-Z_{0}\mathbf{I}\right)$,
where $Z_0$ refers to the reference impedance used to compute the scattering matrix $\boldsymbol{\Theta}$.
In this work, we assume that the antennas at the RIS, at the receiver, and at the transmitter are all matched to this reference impedance, which is set to $Z_0=50$ $\Omega$.
The $N_{I}$-port reconfigurable impedance network is constructed with passive elements which can be adapted to properly reflect the incident signal.
To maximize the power reflected by the RIS, $\mathbf{Z}_{I}$ must be purely reactive and we can write $\mathbf{Z}_{I}=j\mathbf{X}_{I}$, where $\mathbf{X}_{I}\in\mathbb{R}^{N_{I}\times N_{I}}$ denotes the reactance matrix of the $N_{I}$-port reconfigurable impedance network.
Hence, $\boldsymbol{\Theta}$ is given by
\begin{equation}
\boldsymbol{\Theta}=\left(j\mathbf{X}_{I}+Z_{0}\mathbf{I}\right)^{-1}\left(j\mathbf{X}_{I}-Z_{0}\mathbf{I}\right).\label{eq:T(X)}
\end{equation}
Furthermore, the reconfigurable impedance network is also reciprocal so that we have $\mathbf{X}_{I}=\mathbf{X}_{I}^{T}$ and $\boldsymbol{\Theta}=\boldsymbol{\Theta}^{T}$.
Depending on the topology of the reconfigurable impedance network, three different BD-RIS architectures have been identified in \cite{she20}, which are briefly reviewed in the following.
For more detailed information on these novel BD-RIS architectures, we refer the interested reader to \cite{she20}.

\subsection{Single Connected RIS Architecture}

The single connected RIS architecture is the conventional architecture adopted in the literature \cite{wu19a,bas19}.
Here, each port of the reconfigurable impedance network is connected to ground with a reconfigurable impedance and is not connected to the other ports.
The reactance matrix $\mathbf{X}_{I}$ is a diagonal matrix given by $\mathbf{X}_{I}=\mathrm{diag}\left(X_{1},X_{2},\ldots,X_{N_{I}}\right)$, where $X_{n_{I}}$ is the reactance connecting the $n_{I}$th port to ground, for $n_{I}=1,\ldots,N_{I}$.
According to \eqref{eq:T(X)}, the scattering matrix $\boldsymbol{\Theta}$ is also a diagonal matrix written as
\begin{equation}
\boldsymbol{\Theta}=\mathrm{diag}\left(e^{j\theta_{1}},\ldots,e^{j\theta_{N_{I}}}\right),\label{eq:diag(T)}
\end{equation}
where $e^{j\theta_{n_{I}}}=\frac{jX_{n_{I}}-Z_{0}}{jX_{n_{I}}+Z_{0}}$ is the reflection coefficient of the reactance $X_{n_{I}}$, for $n_{I}=1,\ldots,N_{I}$.

\subsection{Fully Connected RIS Architecture}

The fully connected RIS architecture is obtained by connecting every port of the reconfigurable impedance network to all other ports \cite{she20}.
Therefore, the reactance matrix $\mathbf{X}_{I}$ is not restricted to be diagonal and can be any arbitrary symmetric matrix.
According to \eqref{eq:T(X)}, $\boldsymbol{\Theta}$ is a complex symmetric unitary matrix
\begin{equation}
\boldsymbol{\Theta}=\boldsymbol{\Theta}^{T},\:\boldsymbol{\Theta}^{H}\boldsymbol{\Theta}=\boldsymbol{\mathrm{I}}.\label{eq:T fully}
\end{equation}
Comparing the constraints \eqref{eq:diag(T)} and \eqref{eq:T fully}, we notice that the fully connected architecture is a generalization of the single connected one.
More precisely, the scattering matrix $\boldsymbol{\Theta}$ of the fully connected architecture is not limited to being diagonal with unit-modulus entries because of the presence of reconfigurable impedances connecting the ports to each other.
On the one hand, the single connected architecture is the simplest and the most limited since $\boldsymbol{\Theta}$ is a diagonal matrix.
On the other hand, the fully connected architecture offers the highest flexibility at the cost of a higher circuit and optimization complexity.
Graphical representations of single and fully connected RIS architectures can be found in \cite[Fig.~2]{she20}.

\subsection{Group Connected RIS Architecture}

The group connected RIS architecture has been proposed as a trade-off between the single connected and the fully connected to achieve a good balance between performance and complexity \cite{she20}.
In the group connected architecture, the $N_{I}$ elements are divided into $G$ groups, each having $N_{G}=\frac{N_{I}}{G}$ elements.
Each element is connected to all other elements in its group, while there is no connection inter-group. 
Thus, $\mathbf{X}_{I}$ is a block diagonal matrix given by
\begin{equation}
\mathbf{X}_{I}=\mathrm{diag}\left(\mathbf{X}_{I,1},\ldots,\mathbf{X}_{I,G}\right),\:\mathbf{X}_{I,g}=\mathbf{X}_{I,g}^{T},\:\forall g,\label{eq:X group}
\end{equation}
where $\mathbf{X}_{I,g}\in\mathbb{R}^{N_{G}\times N_{G}}$ is the reactance matrix of the $N_{G}$-port fully connected reconfigurable impedance network for the $g$th group.
According to \eqref{eq:T(X)}, the following constraints can be found for the scattering matrix in the group connected architecture
\begin{equation}
\boldsymbol{\Theta}=\mathrm{diag}\left(\boldsymbol{\Theta}_{1},\ldots,\boldsymbol{\Theta}_{G}\right),\:\boldsymbol{\Theta}_{g}=\boldsymbol{\Theta}_{g}^{T},\:\boldsymbol{\Theta}_{g}^{H}\boldsymbol{\Theta}_{g}=\boldsymbol{\mathrm{I}},\:\forall g,\label{eq:T group}
\end{equation}
which show that $\boldsymbol{\Theta}$ is a block diagonal matrix with each block $\boldsymbol{\Theta}_{g}$ being a complex symmetric unitary matrix, for $g=1,\ldots,G$.
Note that the single and fully connected architectures can be seen as two special cases of the group connected architecture, with $N_G=1$ and $N_G=N_I$, respectively \cite{she20}.
Representations of group connected RIS architectures can be found in \cite[Fig.~3]{she20}.

\section{Optimal Design for BD-RIS-Aided Single-User SISO Systems: Reflective Mode}
\label{sec:design}

In this section, the BD-RIS is assumed to work in reflective mode, as typically considered in the literature 
\cite{wu19b,liu21,guo20,li21,fen22,zha22a,li20,zhe21,mei22,ozd20,tan21,wu19c,di20,zhe20,wei21,guo22,abe20,cai20,dai20,dun20}.
Our goal is to design $\boldsymbol{\Theta}$ for group and fully connected RISs to maximize the received signal power given by $P_R=P_T\left\vert h_{RT}+\mathbf{h}_{RI}\boldsymbol{\Theta}\mathbf{h}_{IT}\right\vert^{2}$.
Thus, the optimization problem for group connected RISs writes as
\begin{align}
\underset{\boldsymbol{\Theta}}{\mathsf{\mathrm{max}}}\;\;
&P_T\left\vert h_{RT}+\mathbf{h}_{RI}\boldsymbol{\Theta}\mathbf{h}_{IT}\right\vert^{2}\label{eq:PR}\\
\mathsf{\mathrm{s.t.}}\;\;\;
&\boldsymbol{\Theta}=\mathrm{diag}\left(\boldsymbol{\Theta}_{1},\ldots,\boldsymbol{\Theta}_{G}\right),\label{eq:P-GC-c1}\\
&\boldsymbol{\Theta}_{g}=\boldsymbol{\Theta}_{g}^{T},\:\boldsymbol{\Theta}_{g}^{H}\boldsymbol{\Theta}_{g}=\boldsymbol{\mathrm{I}},\:\forall g,\label{eq:P-GC-c2}
\end{align}
where $P_T=\mathrm{E}[\left|x\right|^{2}]$ is the transmitted signal power.
Since fully connected RISs can be viewed as a special case of group connected RISs, the analogous problem for fully connected RISs can be readily obtained by setting $G=1$ in \eqref{eq:P-GC-c1}.
In a practical development, $\boldsymbol{\Theta}$ can assume a finite number of discretized values due to hardware constraints \cite{ner21}.
However, this is beyond the scope of this paper, where the constraint considered for group connected architectures are given in \eqref{eq:T group}.
Since the entries of $\boldsymbol{\Theta}$ are not constrained to assume discrete values, the term $\mathbf{h}_{RI}\boldsymbol{\Theta}\mathbf{h}_{IT}$ can be always combined in phase with $h_{RT}$.
Thus, we first maximize \eqref{eq:PR} by omitting $h_{RT}$ and then we adjust the phase of the resulting $\boldsymbol{\Theta}$ depending on $\arg\left(h_{RT}\right)$.
We assume unit $P_T$ and introduce the normalized channels $\hat{\mathbf{h}}_{RI}=\mathbf{h}_{RI}/\left\Vert\mathbf{h}_{RI}\right\Vert$ and $\hat{\mathbf{h}}_{IT}=\mathbf{h}_{IT}/\left\Vert\mathbf{h}_{IT}\right\Vert$ such that our problem becomes to maximize
\begin{equation}
\hat{P}_{R}=|\hat{\mathbf{h}}_{RI}\boldsymbol{\Theta}\hat{\mathbf{h}}_{IT}|^2.\label{eq:PR-hat}
\end{equation}

For traditional single connected architectures, it is known that the scattering matrix can be simply optimized in closed-form.
The maximum normalized received signal power is
\begin{equation}
\hat{P}_{R}^{\mathrm{Single}}=\left(\sum_{n_{I}=1}^{N_{I}}\left|[\hat{\mathbf{h}}_{RI}]_{n_{I}}[\hat{\mathbf{h}}_{IT}]_{n_{I}}\right|\right)^{2},
\end{equation}
which is achieved by designing $\boldsymbol{\Theta}$ as in \eqref{eq:diag(T)} with
\begin{equation}
\theta_{n_{I}}=-\arg\left([\hat{\mathbf{h}}_{RI}]_{n_I}\right)-\arg\left([\hat{\mathbf{h}}_{IT}]_{n_I}\right),\forall n_{I}.\label{eq:theta}
\end{equation}
However, an exact solution for the optimal $\boldsymbol{\Theta}$ in group and fully connected RISs is not known given the non-convexity of problem \eqref{eq:PR}-\eqref{eq:P-GC-c2}.
In the following, we consider the global optimization of fully connected RISs in Section~\ref{sec:design}~A.
To this end, we provide an upper bound on the objective function \eqref{eq:PR-hat} and a necessary and sufficient condition to achieve it.
This condition is subsequently transformed into an equivalent condition, for which a closed-form solution is available.
We generalize our approach to group connected RISs in Section~\ref{sec:design}~B.

\subsection{Closed-Form Solution for Optimal Fully Connected RIS}
\label{subsec:fully-connected}

We begin by observing that constraints \eqref{eq:P-GC-c1} and \eqref{eq:P-GC-c2} are equivalent to \eqref{eq:T(X)} and \eqref{eq:X group}.
Since $\mathbf{X}_{I}$ is real symmetric, we can use the eigenvalue decomposition to write $\mathbf{X}_{I}=\mathbf{V}\boldsymbol{\Lambda}\mathbf{V}^T$, where $\boldsymbol{\Lambda}=\text{diag}\left(\lambda_1,\ldots,\lambda_{N_I}\right)\in\mathbb{R}^{N_I\times N_I}$ is a diagonal matrix containing the eigenvalues of $\mathbf{X}_{I}$ ordered in decreasing order and $\mathbf{V}\in\mathbb{R}^{N_I\times N_I}$ is orthonormal.
Applying \eqref{eq:T(X)}, the scattering matrix $\boldsymbol{\Theta}$ is given by
\begin{equation}
\boldsymbol{\Theta}=\left(j\mathbf{V}\boldsymbol{\Lambda}\mathbf{V}^T+Z_{0}\mathbf{I}\right)^{-1}\left(j\mathbf{V}\boldsymbol{\Lambda}\mathbf{V}^T-Z_{0}\mathbf{I}\right)=\mathbf{V}\mathbf{D}\mathbf{V}^T,\label{eq:T-dec}
\end{equation}
where $\mathbf{D}=\text{diag}\left(e^{jd_{1}},\ldots,e^{jd_{N_I}}\right)\in\mathbb{C}^{N_I\times N_I}$ is a diagonal matrix with $e^{jd_{n_I}}=\frac{j\lambda_{n_I}-Z_{0}}{j\lambda_{n_I}+Z_{0}}$.
Note that the complex diagonal elements of the matrix $\mathbf{D}$ have unit modulus by construction.
Using the decomposition of $\boldsymbol{\Theta}$ given by \eqref{eq:T-dec}, the normalized received signal power $\hat{P}_{R}$ in \eqref{eq:PR-hat} can be expressed as
\begin{equation}
\hat{P}_{R}=|\hat{\mathbf{h}}_{RI}\mathbf{V}\mathbf{D}\mathbf{V}^T\hat{\mathbf{h}}_{IT}|^2=\left|\bar{\mathbf{h}}_{RI}\mathbf{D}\bar{\mathbf{h}}_{IT}\right|^2,\label{eq:PR-hat-decomposition}
\end{equation}
where $\bar{\mathbf{h}}_{RI}=\hat{\mathbf{h}}_{RI}\mathbf{V}$ and $\bar{\mathbf{h}}_{IT}=\mathbf{V}^T\hat{\mathbf{h}}_{IT}$.
Note that \eqref{eq:PR-hat-decomposition} is the squared modulus of the dot product between $\bar{\mathbf{h}}_{RI}$ and $\mathbf{D}\bar{\mathbf{h}}_{IT}$.
Thus, using the Cauchy-Schwarz inequality, we have
\begin{equation}
\hat{P}_{R}\leq\left\Vert\bar{\mathbf{h}}_{RI}\right\Vert^2\left\Vert\mathbf{D}\bar{\mathbf{h}}_{IT}\right\Vert^2=1,
\end{equation}
where the equality $\hat{P}_{R}=1$ is achieved if and only if
\begin{equation}
\left|\left[\bar{\mathbf{h}}_{RI}\right]_{n_I}\right|=\left|\left[\bar{\mathbf{h}}_{IT}\right]_{n_I}\right|,\:\forall n_I.\label{eq:opt-single}
\end{equation}
Since we are interested in achieving the received signal power upper bound, our goal is now to find a real orthonormal matrix $\mathbf{V}=\left[\mathbf{v}_1,\ldots,\mathbf{v}_{N_I}\right]$ such that condition \eqref{eq:opt-single} is satisfied.

It is easy to recognize that if the channels $\mathbf{h}_{RI}$ and $\mathbf{h}_{IT}$ are linearly dependent, the optimal $\mathbf{V}$ is $\mathbf{V}=\mathbf{I}$.
Consequently, $\mathbf{D}$ can be designed according to \eqref{eq:theta} and the matrix $\boldsymbol{\Theta}=\mathbf{V}\mathbf{D}\mathbf{V}^T$ is readily obtained.
For this reason, in the following discussion, we assume that the channels $\mathbf{h}_{RI}$ and $\mathbf{h}_{IT}$ are linearly independent.

Our objective is now to transform the optimality condition \eqref{eq:opt-single} into an equivalent condition for which a closed-form solution can be derived.
Noting that $\left[\bar{\mathbf{h}}_{RI}\right]_{n_I}=\hat{\mathbf{h}}_{RI}\mathbf{v}_{n_I}$ and $\left[\bar{\mathbf{h}}_{IT}\right]_{n_I}=\mathbf{v}_{n_I}^T\hat{\mathbf{h}}_{IT}$, condition \eqref{eq:opt-single} becomes equivalent to
\begin{equation}
\left|\hat{\mathbf{h}}_{RI}\mathbf{v}_{n_I}\right|^2=\left|\mathbf{v}_{n_I}^T\hat{\mathbf{h}}_{IT}\right|^2,
\end{equation}
which can be in turn rewritten as
\begin{equation}
\mathbf{v}_{n_I}^T\mathbf{R}_{RI}\mathbf{v}_{n_I}=\mathbf{v}_{n_I}^T\mathbf{R}_{IT}\mathbf{v}_{n_I},\label{eq:quad-forms}
\end{equation}
where $\mathbf{R}_{RI}=\hat{\mathbf{h}}_{RI}^H\hat{\mathbf{h}}_{RI}\in\mathbb{C}^{N_I\times N_I}$ and $\mathbf{R}_{IT}=\hat{\mathbf{h}}_{IT}\hat{\mathbf{h}}_{IT}^H\in\mathbb{C}^{N_I\times N_I}$.
The left- and right-hand sides of \eqref{eq:quad-forms} are quadratic forms.
Since $\mathbf{v}_{n_I}^T\mathbf{R}_{RI}\mathbf{v}_{n_I}=\mathbf{v}_{n_I}^T\mathbf{R}_{RI}^T\mathbf{v}_{n_I}$ and $\mathbf{v}_{n_I}^T\mathbf{R}_{IT}\mathbf{v}_{n_I}=\mathbf{v}_{n_I}^T\mathbf{R}_{IT}^T\mathbf{v}_{n_I}$, we can replace in \eqref{eq:quad-forms} the matrices $\mathbf{R}_{RI}$ and $\mathbf{R}_{IT}$ with their symmetric parts $\mathbf{A}_{RI}=1/2\left(\mathbf{R}_{RI}+\mathbf{R}_{RI}^T\right)\in\mathbb{R}^{N_I\times N_I}$ and $\mathbf{A}_{IT}=1/2\left(\mathbf{R}_{IT}+\mathbf{R}_{IT}^T\right)\in\mathbb{R}^{N_I\times N_I}$, respectively, without changing the two quadratic forms.
Thus, \eqref{eq:quad-forms} is equivalent to
\begin{equation}
\mathbf{v}_{n_I}^T\mathbf{A}_{RI}\mathbf{v}_{n_I}=\mathbf{v}_{n_I}^T\mathbf{A}_{IT}\mathbf{v}_{n_I},
\end{equation}
which in turn becomes
\begin{equation}
\mathbf{v}_{n_I}^T\mathbf{A}\mathbf{v}_{n_I}=0,\label{eq:quad-form}
\end{equation}
where the symmetric matrix $\mathbf{A}=\mathbf{A}_{RI}-\mathbf{A}_{IT}\in\mathbb{R}^{N_I\times N_I}$ has been introduced.
To solve \eqref{eq:quad-form}, let us consider the eigenvalue decomposition $\mathbf{A}=\mathbf{U}\mathbf{\Delta}\mathbf{U}^T$, where $\mathbf{\Delta}=\text{diag}\left(\delta_1,\ldots,\delta_{N_I}\right)\in\mathbb{R}^{N_I\times N_I}$ is a diagonal matrix containing the eigenvalues of $\mathbf{A}$ ordered in decreasing order and $\mathbf{U}\in\mathbb{R}^{N_I\times N_I}$ is orthonormal.
By introducing the orthonormal vectors $\mathbf{t}_{n_I}=\mathbf{U}^T\mathbf{v}_{n_I}\in\mathbb{R}^{N_I\times 1}$, for $n_I=1,\ldots,N_I$, \eqref{eq:quad-form} can be reformulated as a diagonal quadratic form
\begin{equation}
\mathbf{t}_{n_I}^T\mathbf{\Delta}\mathbf{t}_{n_I}=0.\label{eq:diag-quad-form1}
\end{equation}
In other words, we need to solve $\mathbf{s}_{n_I}\boldsymbol{\delta}=0$, where $\mathbf{s}_{n_I}=[\left[\mathbf{t}_{n_I}\right]_1^2,\ldots,\left[\mathbf{t}_{n_I}\right]_{N_I}^2]\in\mathbb{R}^{1\times N_I}$ and $\boldsymbol{\delta}=[\delta_1,\ldots,\delta_{N_I}]^T\in\mathbb{R}^{N_I\times 1}$.
Note that we need to find $N_I$ orthonormal vectors $\mathbf{t}_{n_I}$ which are solutions of \eqref{eq:diag-quad-form1}.
This task is hard in general since the solution space of \eqref{eq:diag-quad-form1} is a non-linear space.
However, in our case, we can rely on the special structure of the vector $\boldsymbol{\delta}$.
As proved in the following, $\boldsymbol{\delta}$ contains only two, three, and four non-zero elements when $N_I=2$, $N_I=3$, and $N_I\geq 4$, respectively.
Thus, we can solve \eqref{eq:diag-quad-form1} in closed-form by separately studying these three cases.
The following proposition is introduced to simplify \eqref{eq:diag-quad-form1} in the cases $N_I\in\left\{2,3\right\}$.
\begin{proposition}
For any linearly independent $\mathbf{h}_{RI}\in\mathbb{C}^{1\times N_I}$ and $\mathbf{h}_{IT}\in\mathbb{C}^{N_I\times 1}$, with $N_I\in\left\{2,3\right\}$, the matrix $\mathbf{A}$ has rank $r\left(\mathbf{A}\right)=N_I$ and trace $\textup{Tr}\left(\mathbf{A}\right)=0$.
\label{pro:NI23}
\end{proposition}
\begin{proof}
Please refer to Appendix~A.
\end{proof}

\subsubsection{$N_I=2$}
In the case of fully connected RISs with $N_I=2$, $\mathbf{A}$ has two eigenvalues, both non-zero and one opposite of the other, as a consequence of Proposition~\ref{pro:NI23}.
Denoting the two eigenvalues of $\mathbf{A}$ as $\delta_1$ and $\delta_2$, we have that the vector $\boldsymbol{\delta}$ writes as $\boldsymbol{\delta}=\left[\delta_1,\delta_2\right]^T$, where $\delta_2=-\delta_1$.
Applying Proposition~\ref{pro:NI23}, we simplify \eqref{eq:diag-quad-form1} as
\begin{equation}
\delta_1\left[\mathbf{t}_{n_I}\right]_1^2-\delta_1\left[\mathbf{t}_{n_I}\right]_2^2=0.\label{eq:diag-quad-form2-NI2}
\end{equation}
Thus, we need to solve
\begin{equation}
\begin{cases}
\left[\mathbf{t}_{n_I}\right]_1^2-\left[\mathbf{t}_{n_I}\right]_{2}^2=0\\
\left[\mathbf{t}_{n_I}\right]_1^2+\left[\mathbf{t}_{n_I}\right]_{2}^2=1
\end{cases},
\end{equation}
where the first equation is derived from \eqref{eq:diag-quad-form2-NI2} and the second equation is the unitary norm constraint on $\mathbf{t}_{n_I}$, $\forall n_I\in\{1,2\}$.
Solving by substitution, we obtain $\left[\mathbf{t}_{n_I}\right]_1^2=\left[\mathbf{t}_{n_I}\right]_2^2=1/2$.
Finally, we choose $\mathbf{t}_1=[\sqrt{1/2},\sqrt{1/2}]^T$ and $\mathbf{t}_2=[\sqrt{1/2},-\sqrt{1/2}]^T$ to guarantee orthonormality.

\subsubsection{$N_I=3$}
Considering fully connected RISs with $N_I=3$, we still rely on proposition~\ref{pro:NI23}  to simplify \eqref{eq:diag-quad-form1}.
As a consequence of Proposition~\ref{pro:NI23}, $\mathbf{A}$ has three eigenvalues, all non-zero.
Denoting the three eigenvalues of $\mathbf{A}$ as $\delta_1$, $\delta_2$, and $\delta_3$, we have that the vector $\boldsymbol{\delta}$ writes as $\boldsymbol{\delta}=\left[\delta_1,\delta_2,\delta_3\right]^T$.
Applying Proposition~\ref{pro:NI23}, we simplify \eqref{eq:diag-quad-form1} as
\begin{equation}
\delta_1\left[\mathbf{t}_{n_I}\right]_1^2+\delta_2\left[\mathbf{t}_{n_I}\right]_2^2+\delta_3\left[\mathbf{t}_{n_I}\right]_3^2=0.\label{eq:diag-quad-form2-NI3}
\end{equation}
We choose the vector $\mathbf{t}_1$ with only the first and the third entries non-zero.
Such a vector always exists since Proposition~\ref{pro:NI23} implies $\delta_1>0$ and $\delta_3<0$.
Thus, we need to solve
\begin{equation}
\begin{cases}
\delta_1\left[\mathbf{t}_{1}\right]_1^2+\delta_{3}\left[\mathbf{t}_{1}\right]_{3}^2=0\\
\left[\mathbf{t}_{1}\right]_1^2+\left[\mathbf{t}_{1}\right]_{3}^2=1
\end{cases},
\end{equation}
where the first equation is derived from \eqref{eq:diag-quad-form2-NI3} and the second equation is the unitary norm constraint.
Solving by substitution, we obtain
\begin{equation}
\begin{cases}
\left[\mathbf{t}_{1}\right]_1^2 = \frac{-\delta_{3}}{\delta_1-\delta_{3}}\\
\left[\mathbf{t}_{1}\right]_{3}^2=1-\frac{-\delta_{3}}{\delta_1-\delta_{3}}=\frac{\delta_1}{\delta_1-\delta_{3}}
\end{cases},\label{eq:t1-NI3}
\end{equation}
giving $\mathbf{t}_{1}=\left[\sqrt{\frac{-\delta_{3}}{\delta_1-\delta_{3}}},0,\sqrt{\frac{\delta_1}{\delta_1-\delta_{3}}}\right]^T$.
Now, we select the two remaining vectors in the form
\begin{align}
& \mathbf{t}_2=\left[\frac{1}{K}\sqrt{\frac{\delta_1}{\delta_1-\delta_3}},\sqrt{1-\frac{1}{K^2}},-\frac{1}{K}\sqrt{\frac{-\delta_3}{\delta_1-\delta_3}}\right]^T,\\
& \mathbf{t}_3=\left[-\frac{1}{K}\sqrt{\frac{\delta_1}{\delta_1-\delta_3}},\sqrt{1-\frac{1}{K^2}},\frac{1}{K}\sqrt{\frac{-\delta_3}{\delta_1-\delta_3}}\right]^T,
\end{align}
where $K$ is a positive constant.
It is easy to recognize that $\mathbf{t}_{1}$, $\mathbf{t}_{2}$, and $\mathbf{t}_{3}$ are an orthonormal basis of $\mathbb{R}^3$ for any $K\neq 1$.
Thus, $K$ must be designed such that $\mathbf{t}_{2}$ and $\mathbf{t}_{3}$ satisfy \eqref{eq:diag-quad-form2-NI3}, that is
\begin{equation}
\frac{\delta_1^2}{K^2\left(\delta_1-\delta_3\right)}+\delta_2\left(1-\frac{1}{K^2}\right)-\frac{\delta_3^2}{K^2\left(\delta_1-\delta_3\right)}=0.\label{eq:K}
\end{equation}
Equation \eqref{eq:K} can be simplified by substituting $\delta_2=-\delta_1-\delta_3$, which is always valid according to Proposition~\ref{pro:NI23}. Eventually, \eqref{eq:K} gives $K=\sqrt{2}$, yielding $\mathbf{t}_2=\left[\sqrt{\frac{\delta_1}{2\left(\delta_1-\delta_3\right)}},\sqrt{\frac{1}{2}},-\sqrt{\frac{-\delta_3}{2\left(\delta_1-\delta_3\right)}}\right]^T$ and $\mathbf{t}_3=\left[-\sqrt{\frac{\delta_1}{2\left(\delta_1-\delta_3\right)}},\sqrt{\frac{1}{2}},\sqrt{\frac{-\delta_3}{2\left(\delta_1-\delta_3\right)}}\right]^T$.

\subsubsection{$N_I\geq 4$}
In the case of fully connected RISs with $N_I\geq 4$, we introduce the following proposition to simplify \eqref{eq:diag-quad-form1}.
\begin{proposition}
For any linearly independent $\mathbf{h}_{RI}\in\mathbb{C}^{1\times N_{I}}$ and $\mathbf{h}_{IT}\in\mathbb{C}^{N_{I}\times 1}$, with $N_I\geq4$, the matrix $\mathbf{A}$ has rank $r\left(\mathbf{A}\right)=4$ and trace $\textup{Tr}\left(\mathbf{A}\right)=0$.
Furthermore, among its four non-zero eigenvalues, two are positive and two are negative.
\label{pro:NI4}
\end{proposition}
\begin{proof}
Please refer to Appendix~B.
\end{proof}

Denoting the first two eigenvalues of $\mathbf{A}$ as $\delta_1$ and $\delta_2$, and the last two as $\delta_{N_I-1}$ and $\delta_{N_I}$, we have that the vector $\boldsymbol{\delta}$ writes as $\boldsymbol{\delta}=\left[\delta_1,\delta_2,0,\ldots,0,\delta_{N_I-1},\delta_{N_I}\right]^T$.
Applying Proposition~\ref{pro:NI4}, we simplify \eqref{eq:diag-quad-form1} as
\begin{equation}
\delta_1\left[\mathbf{t}_{n_I}\right]_1^2+\delta_2\left[\mathbf{t}_{n_I}\right]_2^2+\delta_{N_I-1}\left[\mathbf{t}_{n_I}\right]_{N_I-1}^2+\delta_{N_I}\left[\mathbf{t}_{n_I}\right]_{N_I}^2=0.\label{eq:diag-quad-form2}
\end{equation}
We notice that $N_I-4$ orthonormal solutions to \eqref{eq:diag-quad-form2} are given by the vectors $\mathbf{e}_3,\ldots,\mathbf{e}_{N_I-2}$, where $\mathbf{e}_i\in\mathbb{R}^{N_I\times 1}$ denotes the vector with the $i$th entry being $1$ and the others being $0$, for $i=3,\ldots,N_I-2$.
Thus, we now want to find the remaining four orthonormal vectors $\mathbf{t}_1$, $\mathbf{t}_2$, $\mathbf{t}_3$, $\mathbf{t}_4\in\mathbb{R}^{N_I}$ solutions of \eqref{eq:diag-quad-form2}, all orthogonal to $\mathbf{e}_3,\ldots,\mathbf{e}_{N_I-2}$.
To make them orthogonal to $\mathbf{e}_3,\ldots,\mathbf{e}_{N_I-2}$, it is sufficient to set $[\mathbf{t}_i]_{n_I}=0$ for $i=1,2,3,4$ and $n_I=3,\ldots,N_I-2$.

We choose the first vector $\mathbf{t}_1$ with only the first and the $\left(N_I-1\right)$th entries non-zero.
Note that such a vector always exists since $\delta_1>0$ and $\delta_{N_I-1}<0$.
Thus, we need to solve
\begin{equation}
\begin{cases}
\delta_1\left[\mathbf{t}_{1}\right]_1^2+\delta_{N_I-1}\left[\mathbf{t}_{1}\right]_{N_I-1}^2=0\\
\left[\mathbf{t}_{1}\right]_1^2+\left[\mathbf{t}_{1}\right]_{N_I-1}^2=1
\end{cases},
\end{equation}
where the first equation is derived from \eqref{eq:diag-quad-form2} and the second equation is the unitary norm constraint.
Solving by substitution, we obtain
\begin{equation}
\begin{cases}
\left[\mathbf{t}_{1}\right]_1^2 = \frac{-\delta_{N_I-1}}{\delta_1-\delta_{N_I-1}}\\
\left[\mathbf{t}_{1}\right]_{N_I-1}^2=1-\frac{-\delta_{N_I-1}}{\delta_1-\delta_{N_I-1}}=\frac{\delta_1}{\delta_1-\delta_{N_I-1}}
\end{cases},\label{eq:t1}
\end{equation}
which gives $\mathbf{t}_{1}=\left[\sqrt{\frac{-\delta_{N_I-1}}{\delta_1-\delta_{N_I-1}}},0,\ldots,\sqrt{\frac{\delta_1}{\delta_1-\delta_{N_I-1}}},0\right]^T$.
Similarly, we choose the second vector $\mathbf{t}_2$ with only the second and $N_I$th entries non-zero.
Also this vector always exists since $\delta_2>0$ and $\delta_{N_I}<0$.
With a similar procedure, we obtain
\begin{equation}
\begin{cases}
\left[\mathbf{t}_{2}\right]_2^2 = \frac{-\delta_{N_I}}{\delta_2-\delta_{N_I}}\\
\left[\mathbf{t}_{2}\right]_{N_I}^2=1-\frac{-\delta_{N_I}}{\delta_2-\delta_{N_I}}=\frac{\delta_2}{\delta_2-\delta_{N_I}}
\end{cases},\label{eq:t2}
\end{equation}
giving $\mathbf{t}_{2}=\left[0,\sqrt{\frac{-\delta_{N_I}}{\delta_2-\delta_{N_I}}},\ldots,0,\sqrt{\frac{\delta_2}{\delta_2-\delta_{N_I}}}\right]^T$.
Note that these first two vectors $\mathbf{t}_1$ and $\mathbf{t}_2$ are orthonormal by construction.

The remaining two vectors $\mathbf{t}_3$ and $\mathbf{t}_4$ must be linear combinations of a basis of the null space of the matrix $M=\left[\mathbf{t}_1,\mathbf{t}_2\right]^T$.
Such a basis is readily given by two vectors $\mathbf{b}_1$ and $\mathbf{b}_2$ in the form
\begin{align}
& \mathbf{b}_1=[\left[\mathbf{t}_{1}\right]_{N_I-1},0,\ldots,-\left[\mathbf{t}_{1}\right]_1,0]^T,\\
& \mathbf{b}_2=[0,\left[\mathbf{t}_{2}\right]_{N_I},\ldots,0,-\left[\mathbf{t}_{2}\right]_2]^T,
\end{align}
whose $n_I$th entry is zero, for $n_I=3,\ldots,N_I-2$, in addition to the vectors $\mathbf{e}_3,\ldots,\mathbf{e}_{N_I-2}$.
Thus, $\mathbf{t}_3$ and $\mathbf{t}_4$ can be expressed as a generic linear combination $\mathbf{c}=a_1\mathbf{b}_1+a_2\mathbf{b}_2$ given by
\begin{equation}
\mathbf{c}=\left[a_1\left[\mathbf{t}_{1}\right]_{N_I-1},a_2\left[\mathbf{t}_{2}\right]_{N_I},\dots,-a_1\left[\mathbf{t}_{1}\right]_1,-a_2\left[\mathbf{t}_{2}\right]_2\right]^T.
\end{equation}
Now, our objective is to find $a_1$ and $a_2$ satisfying \eqref{eq:diag-quad-form2} and the unitary norm constraint, that is
\begin{equation}
\begin{cases}
\delta_1a_1^2\left[\mathbf{t}_{1}\right]_{N_I-1}^2+\delta_2a_2^2\left[\mathbf{t}_{2}\right]_{N_I}^2+\delta_{N_I-1}a_1^2\left[\mathbf{t}_{1}\right]_1^2+\delta_{N_I}a_2^2\left[\mathbf{t}_{2}\right]_2^2=0\\
a_1^2\left[\mathbf{t}_{1}\right]_{N_I-1}^2+a_2^2\left[\mathbf{t}_{2}\right]_{N_I}^2+a_1^2\left[\mathbf{t}_{1}\right]_1^2+a_2^2\left[\mathbf{t}_{2}\right]_2^2=1
\end{cases}.\label{eq:a1a2}
\end{equation}
Substituting in \eqref{eq:a1a2} the entries of the vectors $\mathbf{t}_1$ and $\mathbf{t}_2$ given by \eqref{eq:t1} and \eqref{eq:t2}, respectively, we obtain
\begin{equation}
\begin{cases}
(\delta_1+\delta_{N_I-1})a_1^2+(\delta_2+\delta_{N_I})a_2^2=0\\
a_1^2+a_2^2=1
\end{cases}.
\end{equation}
Solving by substitution, we have
\begin{equation}
\begin{cases}
a_1^2 = \frac{-\delta_2-\delta_{N_I}}{\delta_1+\delta_{N_I-1}-\delta_2-\delta_{N_I}}\\
a_2^2 = 1-\frac{-\delta_2-\delta_{N_I}}{\delta_1+\delta_{N_I-1}-\delta_2-\delta_{N_I}}=\frac{\delta_1+\delta_{N_I-1}}{\delta_1+\delta_{N_I-1}-\delta_2-\delta_{N_I}}
\end{cases}.
\end{equation}
Note that this always means $a_1^2=a_2^2=1/2$ since Proposition~\ref{pro:NI4} gives $\delta_1+\delta_{N_I-1}=-\delta_2-\delta_{N_I}$.
Finally, we choose $\mathbf{t}_3=\sqrt{1/2}\mathbf{b}_1+\sqrt{1/2}\mathbf{b}_2$ and $\mathbf{t}_4=\sqrt{1/2}\mathbf{b}_1-\sqrt{1/2}\mathbf{b}_2$ to guarantee orthonormality.

In conclusion, we construct an orthonormal matrix $\mathbf{T}\in\mathbb{R}^{N_I\times N_I}$ depending on the number of RIS elements $N_I$.
If $N_I=2$, $\mathbf{T}=\left[\mathbf{t}_1,\mathbf{t}_2\right]$; if $N_I=3$, $\mathbf{T}=\left[\mathbf{t}_1,\mathbf{t}_2,\mathbf{t}_3\right]$; and if $N_I\geq 4$, $\mathbf{T}=\left[\mathbf{t}_1,\mathbf{t}_2,\mathbf{t}_3,\mathbf{t}_4,\mathbf{e}_3,\ldots,\mathbf{e}_{N-2}\right]$.
Note that the columns of $\mathbf{T}$ are orthogonal with each other, have unitary norm, and solve \eqref{eq:diag-quad-form1}.
At this stage, all the building elements of the optimal scattering matrix, denoted as $\bar{\boldsymbol{\Theta}}$, are available.
Applying \eqref{eq:T-dec}, we can write $\bar{\boldsymbol{\Theta}}=\mathbf{V}\mathbf{D}\mathbf{V}^T$, where $\mathbf{V}=\mathbf{U}\mathbf{T}$ by definition of the columns of $\mathbf{T}$, and $\mathbf{D}$ is designed according to \eqref{eq:theta}.
We summarize the steps necessary to build the optimal $\bar{\boldsymbol{\Theta}}$ in Alg.~\ref{alg:theta-siso}.
Note that the solution provided by Alg.~\ref{alg:theta-siso} is proved to be global optimal by the following two facts.
First, it solves \eqref{eq:diag-quad-form1} by construction.
Second, since \eqref{eq:diag-quad-form1} is equivalent to \eqref{eq:opt-single}, it allows to exactly achieve the objective upper bound $\hat{P}_{R}=1$.
To maximize $P_R$ in the presence of the direct link $h_{RT}$, the scattering matrix can be adjusted as
\begin{equation}
\boldsymbol{\Theta}^\star=e^{j\arg\left(h_{RT}\right)}\bar{\boldsymbol{\Theta}},\label{eq:T dir}
\end{equation}
such that the term $\mathbf{h}_{RI}\boldsymbol{\Theta}^\star\mathbf{h}_{IT}$ is made in phase with $h_{RT}$.

\begin{algorithm}[t]
\setstretch{1.2}
\KwIn{$\mathbf{h}_{RI}\in\mathbb{C}^{1\times N_{I}}$, $\mathbf{h}_{IT}\in\mathbb{C}^{N_{I}\times 1}$}
\KwOut{$\bar{\boldsymbol{\Theta}}$}
$\hat{\mathbf{h}}_{RI}=\frac{\mathbf{h}_{RI}}{\left\Vert\mathbf{h}_{RI}\right\Vert}$, $\hat{\mathbf{h}}_{IT}=\frac{\mathbf{h}_{IT}}{\left\Vert\mathbf{h}_{IT}\right\Vert}$\;
$\mathbf{R}_{RI}=\hat{\mathbf{h}}_{RI}^H\hat{\mathbf{h}}_{RI}$, $\mathbf{R}_{IT}=\hat{\mathbf{h}}_{IT}\hat{\mathbf{h}}_{IT}^H$\;
$\mathbf{A}_{RI}=\frac{\mathbf{R}_{RI}+\mathbf{R}_{RI}^T}{2}$, $\mathbf{A}_{IT}=\frac{\mathbf{R}_{IT}+\mathbf{R}_{IT}^T}{2}$\;
$\mathbf{A}\triangleq\mathbf{U}\mathbf{\Delta}\mathbf{U}^T=\mathbf{A}_{RI}-\mathbf{A}_{IT}$\;
$\boldsymbol{\delta}\triangleq\left[\delta_1,\ldots,\delta_{N_I}\right]^T=\text{diag}\left(\mathbf{\Delta}\right)$\;
\uIf{$N_I==2$}{
$\mathbf{T}=\begin{bmatrix}
\sqrt{\frac{1}{2}} & \sqrt{\frac{1}{2}}\\
\sqrt{\frac{1}{2}} & -\sqrt{\frac{1}{2}}
\end{bmatrix}$\;
}
\uElseIf{$N_I==3$}{
$\mathbf{T}=\begin{bmatrix}
\sqrt{\frac{-\delta_3}{\delta_1-\delta_3}} & \sqrt{\frac{\delta_1}{2\left(\delta_1-\delta_3\right)}} & -\sqrt{\frac{\delta_1}{2\left(\delta_1-\delta_3\right)}}\\
0 & \sqrt{\frac{1}{2}} & \sqrt{\frac{1}{2}}\\
\sqrt{\frac{\delta_1}{\delta_1-\delta_3}} & -\sqrt{\frac{-\delta_3}{2\left(\delta_1-\delta_3\right)}} & \sqrt{\frac{-\delta_3}{2\left(\delta_1-\delta_3\right)}}
\end{bmatrix}$\;
}
\Else{
$\mathbf{t}_{1}=\left[\sqrt{\frac{-\delta_{N_I-1}}{\delta_1-\delta_{N_I-1}}},0,\ldots,\sqrt{\frac{\delta_1}{\delta_1-\delta_{N_I-1}}},0\right]^T$\;
$\mathbf{t}_{2}=\left[0,\sqrt{\frac{-\delta_{N_I}}{\delta_2-\delta_{N_I}}},\ldots,0,\sqrt{\frac{\delta_2}{\delta_2-\delta_{N_I}}}\right]^T$\;
$\mathbf{t}_{3}=\frac{1}{\sqrt{2}}\left[\left[\mathbf{t}_{1}\right]_{N_I-1},\left[\mathbf{t}_{2}\right]_{N_I},\ldots,-\left[\mathbf{t}_{1}\right]_1,-\left[\mathbf{t}_{2}\right]_2\right]^T$\;
$\mathbf{t}_{4}=\frac{1}{\sqrt{2}}\left[\left[\mathbf{t}_{1}\right]_{N_I-1},-\left[\mathbf{t}_{2}\right]_{N_I},\ldots,-\left[\mathbf{t}_{1}\right]_1,\left[\mathbf{t}_{2}\right]_2\right]^T$\;
$\mathbf{T}=\left[\mathbf{t}_1,\mathbf{t}_2,\mathbf{t}_3,\mathbf{t}_4,\mathbf{e}_3,\ldots,\mathbf{e}_{N-2}\right]$\;
}
$\mathbf{V}=\mathbf{U}\mathbf{T}$\;
$d_{n_{I}}=-\arg\left([\hat{\mathbf{h}}_{RI}\mathbf{V}]_{n_I}\right)-\arg\left([\mathbf{V}^T\hat{\mathbf{h}}_{IT}]_{n_I}\right),\:\forall n_I$\;
$\mathbf{D}=\text{diag}\left(e^{jd_1},\ldots,e^{jd_{N_I}}\right)$\;
$\bar{\boldsymbol{\Theta}}=\mathbf{V}\mathbf{D}\mathbf{V}^T$\;
\KwRet{$\bar{\boldsymbol{\Theta}}$}
\caption{Optimal fully connected RIS design for single-user SISO systems.}
\label{alg:theta-siso}
\end{algorithm}

\subsection{Closed-Form Solution for Optimal Group Connected RIS}
\label{subsec:group-connected}

Now, we extend our closed-form strategy to design fully connected architectures to group connected ones.
As previously discussed, we initially omit the direct link $h_{RT}$ and assume unitary transmitted signal power.
Thus, the received signal power for group connected architectures writes as 
\begin{equation}
P_R=\left|\sum_{g=1}^G\mathbf{h}_{RI,g}\boldsymbol{\Theta}_g\mathbf{h}_{IT,g}\right|^2\label{eq:group}
\end{equation}
where $\mathbf{h}_{RI}=[\mathbf{h}_{RI,1},\ldots,\mathbf{h}_{RI,G}]$ with $\mathbf{h}_{RI,g}\in\mathbb{C}^{1\times N_{G}}$ and $\mathbf{h}_{IT}=[\mathbf{h}_{IT,1},\ldots,\mathbf{h}_{IT,G}]^{T}$ with $\mathbf{h}_{IT,g}\in\mathbb{C}^{N_{G}\times1}$ \cite{she20}.
It is easy to recognize that \eqref{eq:group} is maximized when the terms $\mathbf{h}_{RI,g}\boldsymbol{\Theta}_g\mathbf{h}_{IT,g}$ are all individually maximized in absolute value and they are all co-phased.
Recalling the constraint on $\boldsymbol{\Theta}_g$ given by \eqref{eq:T group}, the optimal $\boldsymbol{\Theta}_g$ that maximizes $|\mathbf{h}_{RI,g}\boldsymbol{\Theta}_g\mathbf{h}_{IT,g}|$ is given by Alg.~\ref{alg:theta-siso} applied to the truncated channels $\mathbf{h}_{RI,g}$ and $\mathbf{h}_{IT,g}$, $\forall g$.
Note that $\boldsymbol{\Theta}_g$ constructed by Alg.~\ref{alg:theta-siso} ensures that the complex number $\mathbf{h}_{RI,g}\boldsymbol{\Theta}_g\mathbf{h}_{IT,g}$ has phase zero.
Thus, \eqref{eq:group} is maximized when the matrices $\boldsymbol{\Theta}_g$ are constructed by Alg.~\ref{alg:theta-siso} since all the terms $\mathbf{h}_{RI,g}\boldsymbol{\Theta}_g\mathbf{h}_{IT,g}$ are co-phased.
The block diagonal matrix $\boldsymbol{\Theta}$ is finally obtained from the matrices $\boldsymbol{\Theta}_g$ applying \eqref{eq:T group}.
To maximize $P_R$ in the presence of the direct link $h_{RT}$, also for group connected architectures $\boldsymbol{\Theta}$ can be adjusted as in \eqref{eq:T dir}.

\section{Optimal Design for BD-RIS-Aided Single-User SISO Systems: Transmissive Mode}
\label{sec:design-transmissive}

In Section~\ref{sec:design}, we assumed the BD-RIS to work in reflective mode.
This implies that both the transmitter and the receiver are covered by all the RIS elements.
In other words, all the entries of $\mathbf{h}_{RI}$ and $\mathbf{h}_{IT}$ are non-zero in general.
In this section, we study the case in which the BD-RIS works in transmissive mode, as modeled in \cite{li22-1}.
Following \cite{li22-1}, we consider a BD-RIS made of $N_I=2M_I$ elements, where $M_I$ is the number of RIS cells.
Each RIS cell is formed by two RIS elements placed back to back and connected to each other through a reconfigurable impedance.
Specifically, we assume that the $m_I$th cell is formed by the $(2m_I-1)$th and $(2m_I)$th RIS elements, for $m_I=1,\ldots,M_I$.
With this notation, the RIS elements can be partitioned into two sectors, where sector 1 is formed by the odd RIS elements and sector 2 is formed by the even RIS elements.
Thus, the whole space is divided into two sides, respectively covered by the two sectors.
When the BD-RIS is working in transmissive mode, the transmitter and the receiver are located in opposite sectors.
In the following, we assume the transmitter to be in sector 1 and the receiver in sector 2.

Denoting as $\tilde{\mathbf{h}}_{RI}\in\mathbb{C}^{1\times 2M_{I}}$ the channel from the RIS to the receiver, and as $\tilde{\mathbf{h}}_{IT}\in\mathbb{C}^{2M_{I}\times 1}$ the channel from the transmitter to the RIS, the received signal power maximization problem writes as
\begin{align}
\underset{\boldsymbol{\Theta}}{\mathsf{\mathrm{max}}}\;\;
&P_T\left|h_{RT}+\tilde{\mathbf{h}}_{RI}\boldsymbol{\Theta}\tilde{\mathbf{h}}_{IT}\right|^2\label{eq:PR-CW}\\
\mathsf{\mathrm{s.t.}}\;\;\;
&\boldsymbol{\Theta}=\mathrm{diag}\left(\boldsymbol{\Theta}_{1},\ldots,\boldsymbol{\Theta}_{G}\right),\\
&\boldsymbol{\Theta}_{g}=\boldsymbol{\Theta}_{g}^{T},\:\boldsymbol{\Theta}_{g}^{H}\boldsymbol{\Theta}_{g}=\boldsymbol{\mathrm{I}},\:\forall g,\label{eq:PR-CW-c2}
\end{align}
where the odd entries of the channel $\tilde{\mathbf{h}}_{RI}$ are zero, as well as the even entries of the channel $\tilde{\mathbf{h}}_{IT}$.
%
To solve \eqref{eq:PR-CW}-\eqref{eq:PR-CW-c2}, we can readily apply our optimal solution presented in Section~\ref{subsec:group-connected}, with the only difference that half of the entries of $\tilde{\mathbf{h}}_{RI}$ and $\tilde{\mathbf{h}}_{IT}$ are zero when the BD-RIS is working in transmissive mode.
The reader is referred to \cite{li22-1} for detailed information about the optimization of BD-RISs supporting hybrid transmissive and reflective mode in multi-user scenarios.

\section{BD-RIS-Aided Single-User MIMO Systems}
\label{sec:su-mimo}

In this section, we extend our optimal design strategy to single-user \gls{mimo} systems.
W consider an $N_T$ antenna transmitter and an $N_R$ antenna receiver, whose communication is aided by a BD-RIS working in reflective mode.
The equivalent channel writes as $\mathbf{H}=\mathbf{H}_{RT}+\mathbf{H}_{RI}\boldsymbol{\Theta}\mathbf{H}_{IT}$, where $\mathbf{H}_{RT}\in\mathbb{C}^{N_{R}\times N_{T}}$, $\mathbf{H}_{RI}\in\mathbb{C}^{N_{R}\times N_{I}}$, and $\mathbf{H}_{IT}\in\mathbb{C}^{N_{I}\times N_{T}}$ are the channels from the transmitter to receiver, from the RIS to the receiver, and from the transmitter to the RIS, respectively.
Considering single-stream transmission, we denote as $\mathbf{w}\in\mathbb{C}^{N_{T}\times1}$ and $\mathbf{g}\in\mathbb{C}^{1\times N_{R}}$ the precoding and combining vectors, respectively, subject to the constraint $\left\|\mathbf{w}\right\|=1$ and $\left\|\mathbf{g}\right\|=1$.
Thus, the received signal power is $P_{R}^{\mathrm{MIMO}}=P_{T}|\mathbf{g}(\mathbf{H}_{RT}+\mathbf{H}_{RI}\boldsymbol{\Theta}\mathbf{H}_{IT})\mathbf{w}|^{2}$, with corresponding maximization problem
\begin{align}
\underset{\mathbf{w},\mathbf{g},\boldsymbol{\Theta}}{\mathsf{\mathrm{max}}}\;\;
&P_{T}\left|\mathbf{g}\left(\mathbf{H}_{RT}+\mathbf{H}_{RI}\boldsymbol{\Theta}\mathbf{H}_{IT}\right)\mathbf{w}\right|^{2}\label{eq:PR-MIMO}\\
\mathsf{\mathrm{s.t.}}\;\;\;
&\boldsymbol{\Theta}=\mathrm{diag}\left(\boldsymbol{\Theta}_{1},\ldots,\boldsymbol{\Theta}_{G}\right),\\
&\boldsymbol{\Theta}_{g}=\boldsymbol{\Theta}_{g}^{T},\:\boldsymbol{\Theta}_{g}^{H}\boldsymbol{\Theta}_{g}=\boldsymbol{\mathrm{I}},\:\forall g,\\
&\left\|\mathbf{w}\right\|=1,\:\left\|\mathbf{g}\right\|=1,
\end{align}
which is solved by jointly designing $\mathbf{w}$, $\mathbf{g}$, and $\boldsymbol{\Theta}$.

\subsection{Optimizing Fully Connected RIS-Aided Systems Without Direct Link}

We first consider a system aided by a fully connected RIS, and we assume that the direct channel between transmitter and receiver $\mathbf{H}_{RT}\in\mathbb{C}^{N_{R}\times N_{T}}$ is negligible compared to the channel reflected by the RIS.
This assumption reflects real scenarios where the direct channel is highly obstructed and significantly weaker than the RIS-aided link.
In this case, the optimal precoder and combiner are given by the dominant eigenvectors of the equivalent channel $\mathbf{H}_{RI}\boldsymbol{\Theta}\mathbf{H}_{IT}$.
Thus, the maximum received signal power is given by $P_{T}\left\|\mathbf{H}_{RI}\boldsymbol{\Theta}\mathbf{H}_{IT}\right\|^2$, which is upper bounded by
\begin{equation}
\bar{P}_{R}^{\mathrm{MIMO}}=P_{T}\left\Vert\mathbf{H}_{RI}\right\Vert^2\left\Vert\mathbf{H}_{IT}\right\Vert^2,\label{eq:PR-MIMO-UB}
\end{equation}
following the sub-multiplicativity of the spectral norm.
To achieve this upper bound, $\boldsymbol{\Theta}$ must satisfy
\begin{equation}
\mathbf{v}_{RI}=\boldsymbol{\Theta}\mathbf{u}_{IT},\label{eq:opt-cond-u}
\end{equation}
where $\mathbf{v}_{RI}$ is the dominant right singular vector of $\mathbf{H}_{RI}$ and $\mathbf{u}_{IT}$ is the dominant left singular vector of $\mathbf{H}_{IT}$ \cite{ner21}.
Note that the equality \eqref{eq:opt-cond-u} is to be intended up to a phase shift since the complex singular vectors of a matrix are only defined up to a phase shift.
Thus, condition \eqref{eq:opt-cond-u} is satisfied when the cosine similarity
\begin{equation}
\rho=\left\vert\mathbf{v}_{RI}^H\boldsymbol{\Theta}\mathbf{u}_{IT}\right\vert^2\label{eq:rho-MIMO}
\end{equation}
is maximized, i.e., $\rho=1$.
Maximizing \eqref{eq:rho-MIMO} is similar to the problem of maximizing the normalized received signal power in \eqref{eq:PR-hat}, exactly solved for the \gls{siso} setting in Section~\ref{subsec:fully-connected}.
Thus, the optimal $\boldsymbol{\Theta}$ satisfying \eqref{eq:opt-cond-u} can be found by applying Alg.~\ref{alg:theta-siso} to the vectors $\mathbf{v}_{RI}^H$ and $\mathbf{u}_{IT}$ and the upper bound \eqref{eq:PR-MIMO-UB} is tight.
In the \gls{miso} setting, the optimal $\boldsymbol{\Theta}$ is given by Alg.~\ref{alg:theta-siso} applied to the vectors $\mathbf{h}_{RI}$ and $\mathbf{u}_{IT}$.

\subsection{Optimizing Fully/Group Connected RIS-Aided Systems With Direct Link}
\label{subsec:su-mimo-general}

In single-user RIS-aided \gls{mimo} systems, tight upper bounds on the received signal power are not available in general, i.e., when group connected RISs are considered or the direct link is not negligible.
For these cases, we propose a sub-optimal solution to maximize the received signal power in which the matrix $\boldsymbol{\Theta}$ and the beamforming vectors $\mathbf{w}$ and $\mathbf{g}$ are alternatively optimized, as established in the literature on single connected RISs \cite{wu21,wu19b}.
After $\mathbf{w}$ and $\mathbf{g}$ are initialized to feasible values, this optimization process alternates between the two following steps until convergence is reached.
With fixed $\mathbf{w}$ and $\mathbf{g}$, we update $\boldsymbol{\Theta}$ by optimally maximizing the objective $\left|\mathbf{g}\mathbf{H}_{RT}\mathbf{w}+\mathbf{g}\mathbf{H}_{RI}\boldsymbol{\Theta}\mathbf{H}_{IT}\mathbf{w}\right|^2$ as proposed for \gls{siso} systems.
The optimal $\boldsymbol{\Theta}$ is obtained by applying the strategy proposed in Section~\ref{subsec:group-connected} to the channels $h_{RT}^{\textrm{eff}}=\mathbf{g}\mathbf{H}_{RT}\mathbf{w}$, $\mathbf{h}_{RI}^{\textrm{eff}}=\mathbf{g}\mathbf{H}_{RI}$, and $\mathbf{h}_{IT}^{\textrm{eff}}=\mathbf{H}_{IT}\mathbf{w}$.
With fixed $\boldsymbol{\Theta}$, we update $\mathbf{w}$ and $\mathbf{g}$ as the dominant right and left singular vectors of $\mathbf{H}_{RT}+\mathbf{H}_{RI}\boldsymbol{\Theta}\mathbf{H}_{IT}$, respectively.
The convergence is considered reached when the fractional increase of the objective $P_{R}^{\mathrm{MIMO}}$ in a full iteration is below a certain parameter $\epsilon$.

\begin{algorithm}[t]
\setstretch{1.2}
\KwIn{$\mathbf{H}_{RT}\in\mathbb{C}^{N_R\times N_T}$, $\mathbf{H}_{RI}\in\mathbb{C}^{N_R\times N_{I}}$, $\mathbf{H}_{IT}\in\mathbb{C}^{N_{I}\times N_T}$, $N_G$, $\epsilon$}
\KwOut{$\boldsymbol{\Theta}$}
\uIf{$\underbar{P}_R^{\textup{dir}}\geq\underbar{P}_R^{\textup{refl}}$}{
$\mathbf{w}=\mathbf{v}_{RT}$, $\mathbf{g}=\mathbf{u}_{RT}^H$\;
}
\Else{
$\mathbf{w}=\mathbf{v}_{IT}$, $\mathbf{g}=\mathbf{u}_{RI}^H$\;
}
\Repeat{The fractional increase of the objective $P_{R}^{\mathrm{MIMO}}$ is below $\epsilon$}{
$h_{RT}^{\textrm{eff}}=\mathbf{g}\mathbf{H}_{RT}\mathbf{w}$, $\mathbf{h}_{RI}^{\textrm{eff}}=\mathbf{g}\mathbf{H}_{RI}$, $\mathbf{h}_{IT}^{\textrm{eff}}=\mathbf{H}_{IT}\mathbf{w}$\;
Compute $\boldsymbol{\Theta}$ by applying Alg. 1 to $\mathbf{h}_{RI}^{\textrm{eff}}$, $\mathbf{h}_{IT}^{\textrm{eff}}$\;
$\boldsymbol{\Theta}=e^{j\arg\left(h_{RT}^{\textrm{eff}}\right)}\boldsymbol{\Theta}$\;
$\mathbf{w}=\mathbf{v}_{\text{max}}\left(\mathbf{H}_{RT}+\mathbf{H}_{RI}\boldsymbol{\Theta}\mathbf{H}_{IT}\right)$\;
$\mathbf{g}=\mathbf{u}_{\text{max}}^H\left(\mathbf{H}_{RT}+\mathbf{H}_{RI}\boldsymbol{\Theta}\mathbf{H}_{IT}\right)$\;
}
\KwRet{$\boldsymbol{\Theta}$}
\caption{BD-RISs design for single-user MIMO systems.}
\label{alg:theta-mimo}
\end{algorithm}

Depending on the direct link strength, we consider two possible initializations for the beamforming vectors $\mathbf{w}$ and $\mathbf{g}$.
In the following, we define the left and right dominant singular vectors of the matrices $\mathbf{H}_{ij}$ as $\mathbf{u}_{ij}$ and $\mathbf{v}_{ij}$, respectively, for $ij\in\left\{ RT,RI,IT\right\}$.
If the direct link is particularly strong, we set $\mathbf{w}=\mathbf{v}_{RT}$ and $\mathbf{g}=\mathbf{u}_{RT}^H$ to capture the energy in the direct channel dominant eigenmode.
In this case, at the first iteration of the optimization process, it is possible to design $\boldsymbol{\Theta}$ through our optimal solution to achieve a received signal power
\begin{equation}
P_R^{\textrm{dir}}=P_{T}\left(\left\Vert \mathbf{H}_{RT}\right\Vert+\sum_{g=1}^{G}\left\Vert \mathbf{h}_{R,g}\right\Vert\left\Vert\mathbf{h}_{T,g}\right\Vert\right)^2,\label{eq:PR-LB-dir}
\end{equation}
where $\mathbf{h}_{R}=\mathbf{u}_{RT}^H\mathbf{H}_{RI}$ and $\mathbf{h}_{T}=\mathbf{H}_{IT}\mathbf{v}_{RT}$.
Conversely, if the direct link is weak or a high number of RIS elements is employed, we set $\mathbf{w}=\mathbf{v}_{IT}$ and $\mathbf{g}=\mathbf{u}_{RI}^H$ to capture the energy of the reflected link.
In this case, at the first iteration of the optimization process, $\boldsymbol{\Theta}$ can be optimized to achieve a received signal power
%
\begin{multline}
P_R^{\textrm{refl}}=P_{T}\Biggl(\left\vert \mathbf{u}_{RI}^H\mathbf{H}_{RT}\mathbf{v}_{IT}\right\vert\\
+\left\Vert\mathbf{H}_{RI}\right\Vert\left\Vert \mathbf{H}_{IT}\right\Vert\sum_{g=1}^{G}\left\Vert \mathbf{v}_{RI,g}^H\right\Vert\left\Vert\mathbf{u}_{IT,g}\right\Vert\Biggl)^{2}.\label{eq:PR-LB-refl}
\end{multline}
%
Because of the initialization strategy, a lower bound on the received signal power achieved in \gls{mimo} systems is given by $\underbar{P}_R^{\mathrm{MIMO}}=\max\{P_R^{\textrm{dir}},P_R^{\textrm{refl}}\}$, which is the received signal power obtained after the first iteration of our optimization process.
Note that this is a lower bound since the objective function $P_{R}^{\mathrm{MIMO}}$ is non-decreasing over iterations.

We summarize the steps necessary to optimize $\boldsymbol{\Theta}$ in single-user \gls{mimo} systems in Alg.~\ref{alg:theta-mimo}.
The convergence of Alg.~\ref{alg:theta-mimo} is guaranteed by the following two facts.
First, at each iteration, the objective $P_{R}^{\mathrm{MIMO}}$ is non-decreasing.
Second, the objective function is bounded from above by $P_T\left(\left\Vert\mathbf{H}_{RT}\right\Vert+\left\Vert\mathbf{H}_{RI}\right\Vert\left\Vert\mathbf{H}_{IT}\right\Vert\right)^2$ because of the triangle inequality, and the sub-multiplicativity of the spectral norm.
Note that Alg.~\ref{alg:theta-mimo} can be readily applied to the \gls{miso} setting, as it is a special case of the \gls{mimo} setting, in which $N_R=1$.
Single-user \gls{mimo} systems aided by a BD-RIS working in transmissive mode can be similarly optimized by directly applying the discussion made in Section~\ref{sec:design-transmissive}.

\section{BD-RIS-Aided Multi-User MISO Systems}
\label{sec:mu-mimo}

In this section, we study the weighted sum power maximization in multi-user \gls{miso} systems, which is a problem particularly relevant in \gls{wpt} applications \cite{wu20}.
Let us consider an $N_T$ antenna transmitter serving $K$ single-antenna receivers through the support of a BD-RIS working in reflective mode.
We denote the channel from the transmitter to the $k$th receiver and from the RIS to the $k$th receiver as $\mathbf{h}_{RT,k}\in\mathbb{C}^{1\times N_T}$ and $\mathbf{h}_{RI,k}\in\mathbb{C}^{1\times N_I}$, respectively.
Consequently, the equivalent channel seen by the $k$th receiver is denoted as $\mathbf{h}_k=\mathbf{h}_{RT,k}+\mathbf{h}_{RI,k}\boldsymbol{\Theta}\mathbf{H}_{IT}$.

In general, the transmitted signal writes as
\begin{equation}
\mathbf{x}=\sqrt{P_T}\sum_{k=1}^K \mathbf{w}_ks_k,
\end{equation}
where the precoding vectors $\mathbf{w}_k$ are subject to the constraints $\sum_{k=1}^K\left\Vert\mathbf{w}_k\right\Vert^2=1$ and $s_k$ are the energy-carrying signals subject to $\mathrm{E}[\left|s_k\right|^{2}]=1$.
Thus, the received signal power at the $k$th receiver writes as
\begin{equation}
P_{R,k}=P_T\sum_{j=1}^K\left\vert\mathbf{h}_k\mathbf{w}_j\right\vert^2.\label{eq:SR-k}
\end{equation}
Denoting by $\alpha_k>0$ the power weight of the $k$th receiver, the weighted sum power writes as $S_R=\sum_{k=1}^K\alpha_kP_{R,k}$, with corresponding maximization problem given by
\begin{align}
\underset{\mathbf{w}_k,\boldsymbol{\Theta}}{\mathsf{\mathrm{max}}}\;\;
&\sum_{k=1}^K\alpha_kP_{R,k}\label{eq:SR}\\
\mathsf{\mathrm{s.t.}}\;\;\;
&\boldsymbol{\Theta}=\mathrm{diag}\left(\boldsymbol{\Theta}_{1},\ldots,\boldsymbol{\Theta}_{G}\right),\\
&\boldsymbol{\Theta}_{g}=\boldsymbol{\Theta}_{g}^{T},\:\boldsymbol{\Theta}_{g}^{H}\boldsymbol{\Theta}_{g}=\boldsymbol{\mathrm{I}},\:\forall g,\\
&\sum_{k=1}^K\left\|\mathbf{w}_k\right\|^2=1.
\end{align}
Substituting \eqref{eq:SR-k} into \eqref{eq:SR}, we obtain
\begin{equation}
S_R=\sum_{k=1}^K\alpha_kP_T\sum_{j=1}^K\left\vert\mathbf{h}_k\mathbf{w}_j\right\vert^2=\sum_{j=1}^KP_T\mathbf{w}_j^H\mathbf{S}\mathbf{w}_j,\label{eq:SR2}
\end{equation}
where we introduced $\mathbf{S}=\sum_{k=1}^K\alpha_k\mathbf{h}_k^H\mathbf{h}_k$.
From \eqref{eq:SR2}, we notice that the optimal precoding vectors $\mathbf{w}_j$ should be all aligned with the dominant eigenvector of $\mathbf{S}$, denoted as $\mathbf{v}_{\text{max}}\left(\mathbf{S}\right)$.
As in \cite{wu20}, we consider a single-stream precoding given by $\mathbf{w}=\mathbf{v}_{\text{max}}\left(\mathbf{S}\right)$, with no loss of optimality.
With this optimal precoding, the weighted sum power is
\begin{equation}
S_R=P_T\left\Vert\mathbf{S}\right\Vert=P_T\left\Vert\mathbf{H}^H\mathbf{H}\right\Vert=P_T\left\Vert\mathbf{H}\right\Vert^2,\label{eq:SR(H)}
\end{equation}
where we introduced $\mathbf{H}=\left[\sqrt{\alpha_1}\mathbf{h}_1^H,\ldots,\sqrt{\alpha_K}\mathbf{h}_K^H\right]^H$.
To maximize $S_R$, it is convenient to rewrite \eqref{eq:SR(H)} by explicitly highlighting the role of $\boldsymbol{\Theta}$.
Defining the matrices $\mathbf{G}_{RT}=\left[\sqrt{\alpha_1}\mathbf{h}_{RT,1}^H,\ldots,\sqrt{\alpha_K}\mathbf{h}_{RT,K}^H\right]^H$ and $\mathbf{G}_{RI}=\left[\sqrt{\alpha_1}\mathbf{h}_{RI,1}^H,\ldots,\sqrt{\alpha_K}\mathbf{h}_{RI,K}^H\right]^H$, we can write $\mathbf{H}=\mathbf{G}_{RT}+\mathbf{G}_{RI}\boldsymbol{\Theta}\mathbf{H}_{IT}$.

\subsection{Optimizing Fully Connected RIS-Aided Systems Without Direct Links}

\begin{figure*}[t]
\centering
\includegraphics[width=0.42\textwidth]{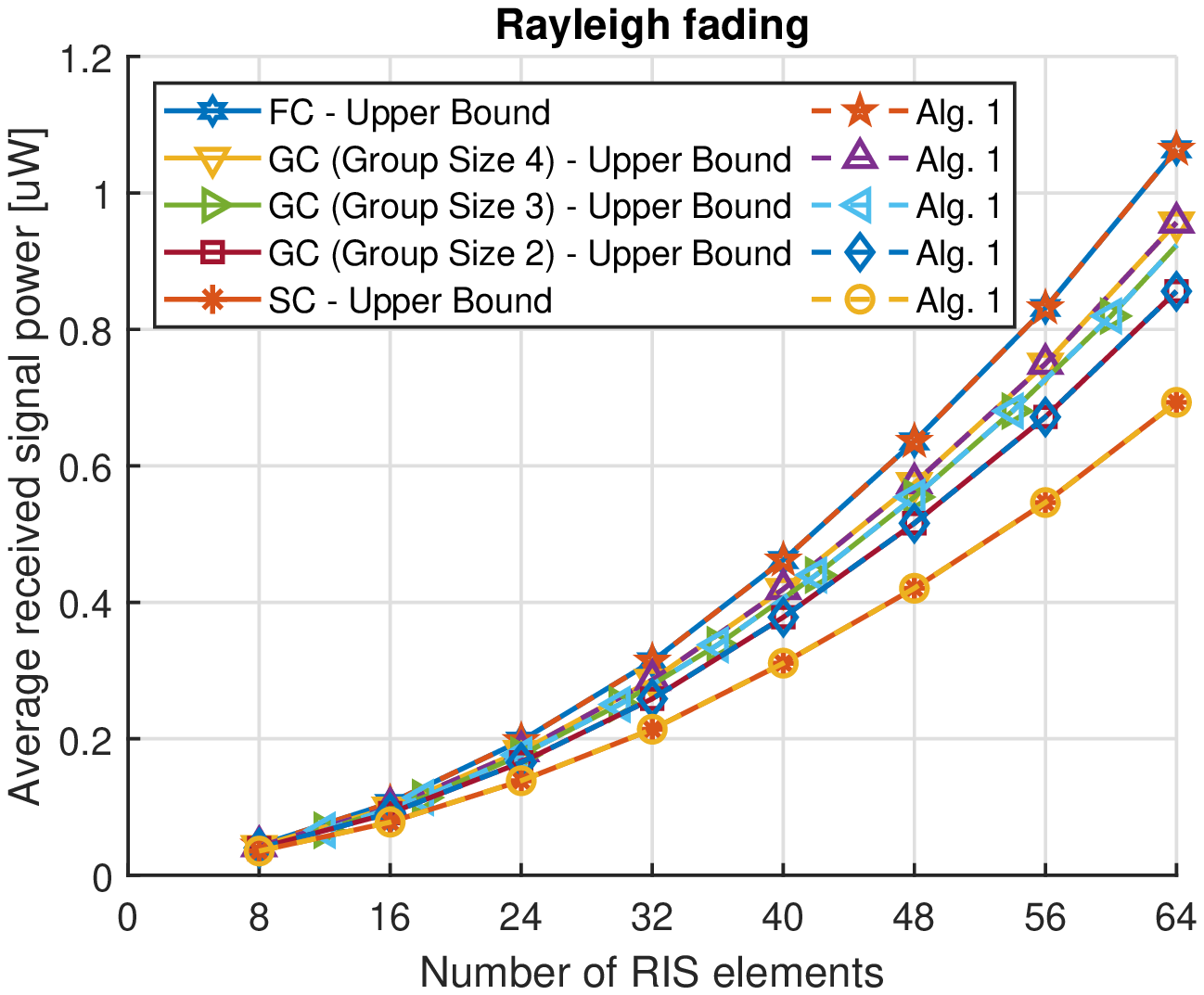}
\includegraphics[width=0.42\textwidth]{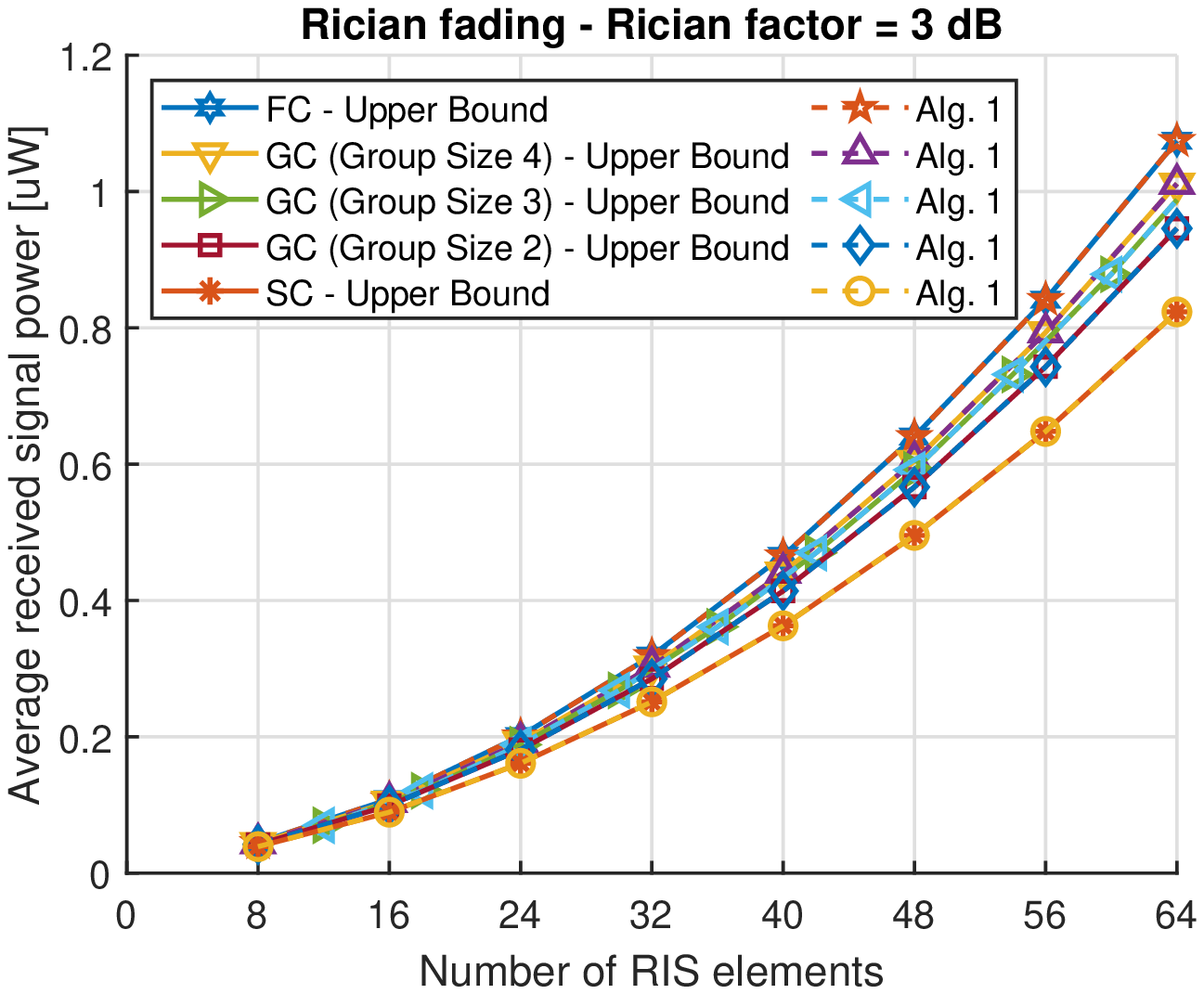}
\caption{Average received signal power in SISO systems aided by single connected ``SC'', group connected ``GC'', and fully connected ``FC'' BD-RISs working in reflective mode.}
\label{fig:pr}
\end{figure*}
\begin{figure*}[t]
\centering
\includegraphics[width=0.42\textwidth]{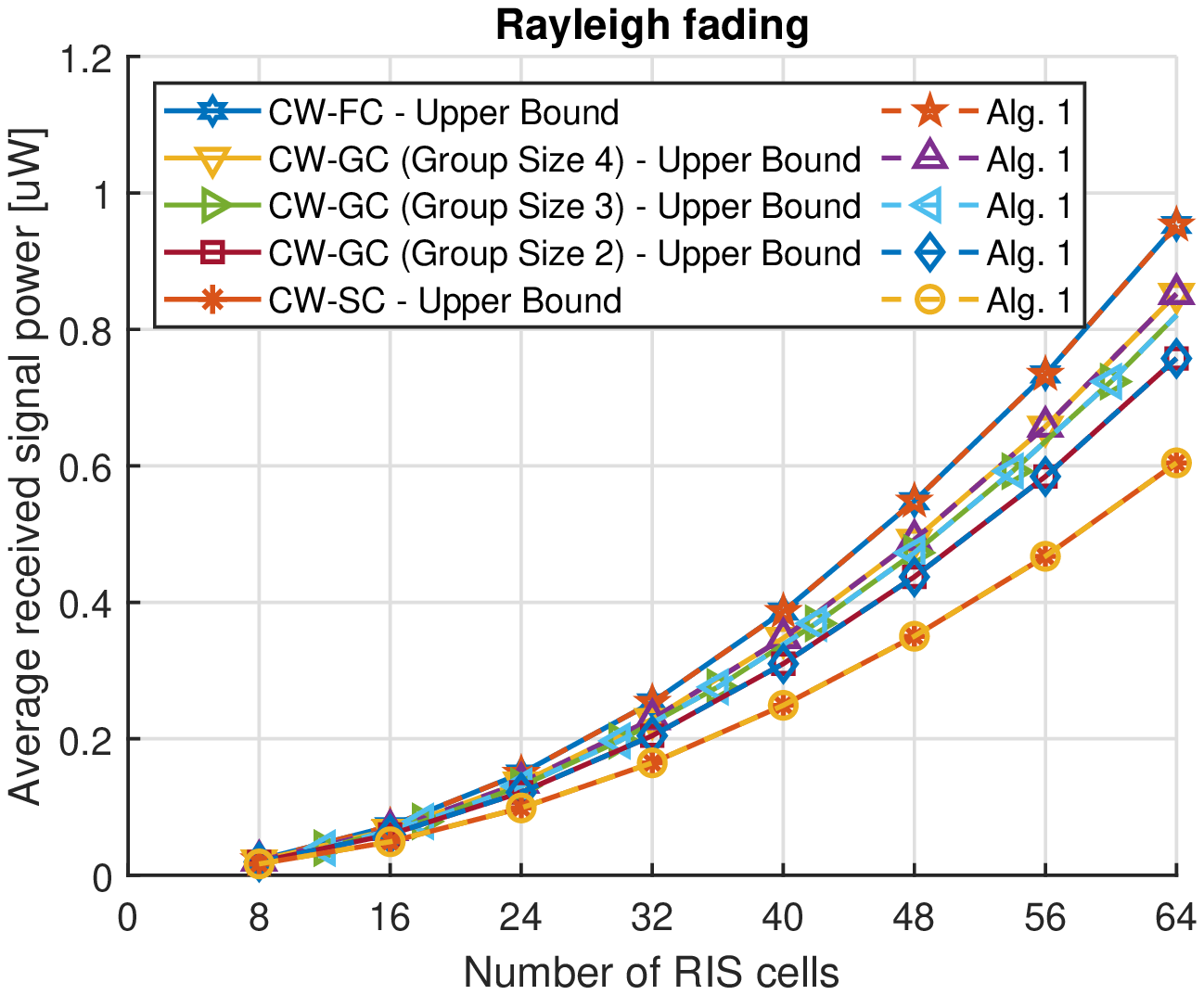}
\includegraphics[width=0.42\textwidth]{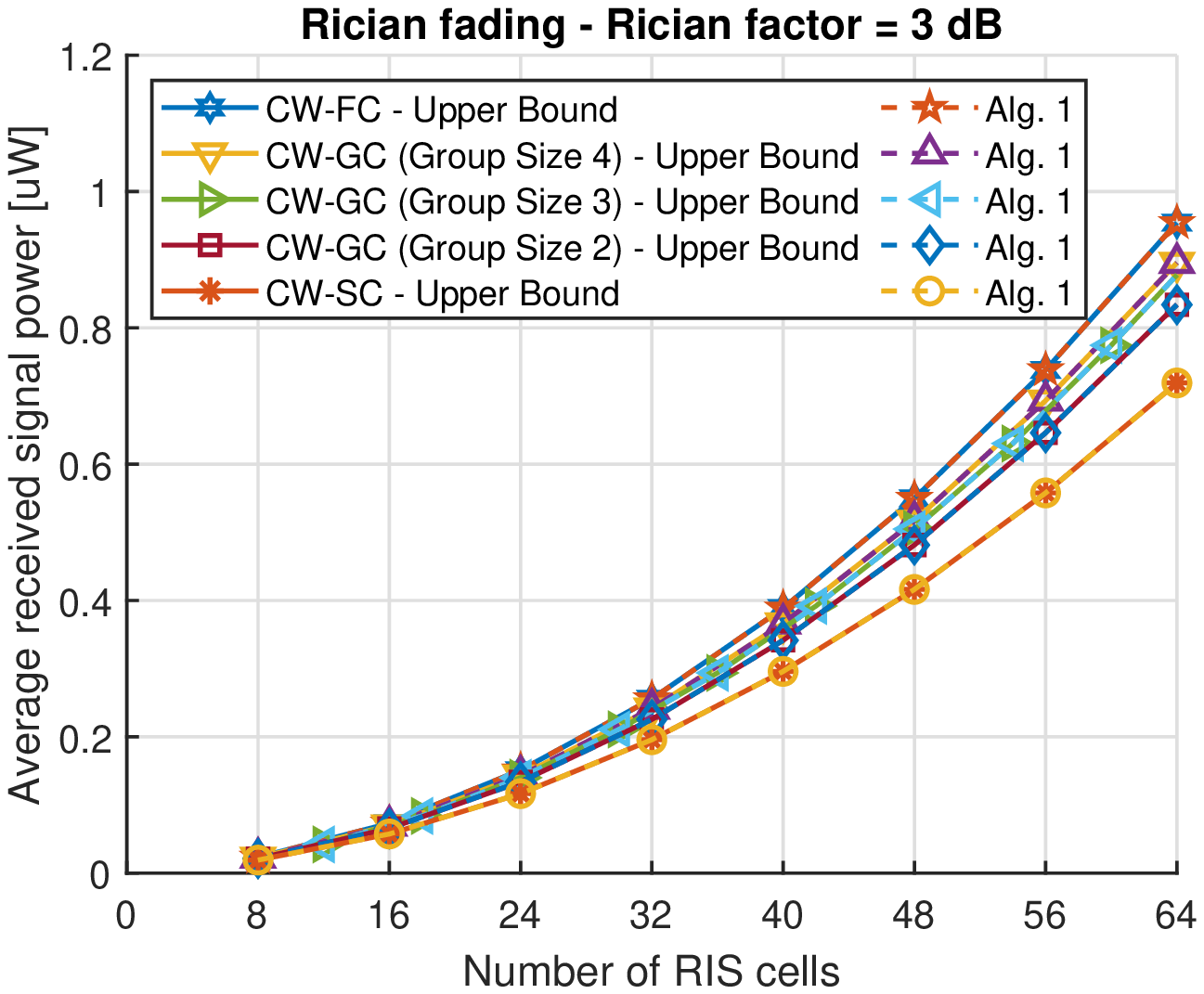}
\caption{Average received signal power in SISO systems aided by cell-wise single connected ``CW-SC'', group connected ``CW-GC'', and fully connected ``CW-FC'' BD-RISs working in transmissive mode.}
\label{fig:pr-transmissive-mode}
\end{figure*}

We first consider a system aided by a fully connected RIS, and we assume that the direct channels between transmitter and receivers are negligible compared to the channels reflected by the RIS.
Consequently, the equivalent channel seen by the $k$th receiver is given by $\mathbf{h}_k=\mathbf{h}_{RI,k}\boldsymbol{\Theta}\mathbf{H}_{IT}$, yielding $\mathbf{H}=\mathbf{G}_{RI}\boldsymbol{\Theta}\mathbf{H}_{IT}$.
Thus, the maximum weighted sum power is given by $P_T\left\Vert\mathbf{G}_{RI}\boldsymbol{\Theta}\mathbf{H}_{IT}\right\Vert^2$, which is upper bounded by
\begin{equation}
\bar{S}_R=P_T\left\Vert\mathbf{G}_{RI}\right\Vert^2\left\Vert\mathbf{H}_{IT}\right\Vert^2,\label{eq:SR-UB}
\end{equation}
because of the sub-multiplicativity of the spectral norm.
Using the discussion carried out for the single-user \gls{mimo} setting, we introduce the vectors $\mathbf{t}_{RI}$ and $\mathbf{u}_{IT}$ as the dominant left singular vectors of $\mathbf{G}_{RI}^{H}$ and $\mathbf{H}_{IT}$, respectively.
Thus, the global optimal $\boldsymbol{\Theta}$ achieving \eqref{eq:SR-UB} can be found by applying Alg.~\ref{alg:theta-siso} to the vectors $\mathbf{t}_{RI}^H$ and $\mathbf{u}_{IT}$.
This proves that the upper bound \eqref{eq:SR-UB} is tight.

\subsection{Optimizing Fully/Group Connected RIS-Aided Systems With Direct Links}

\begin{figure*}[t]
\centering
\includegraphics[width=0.42\textwidth]{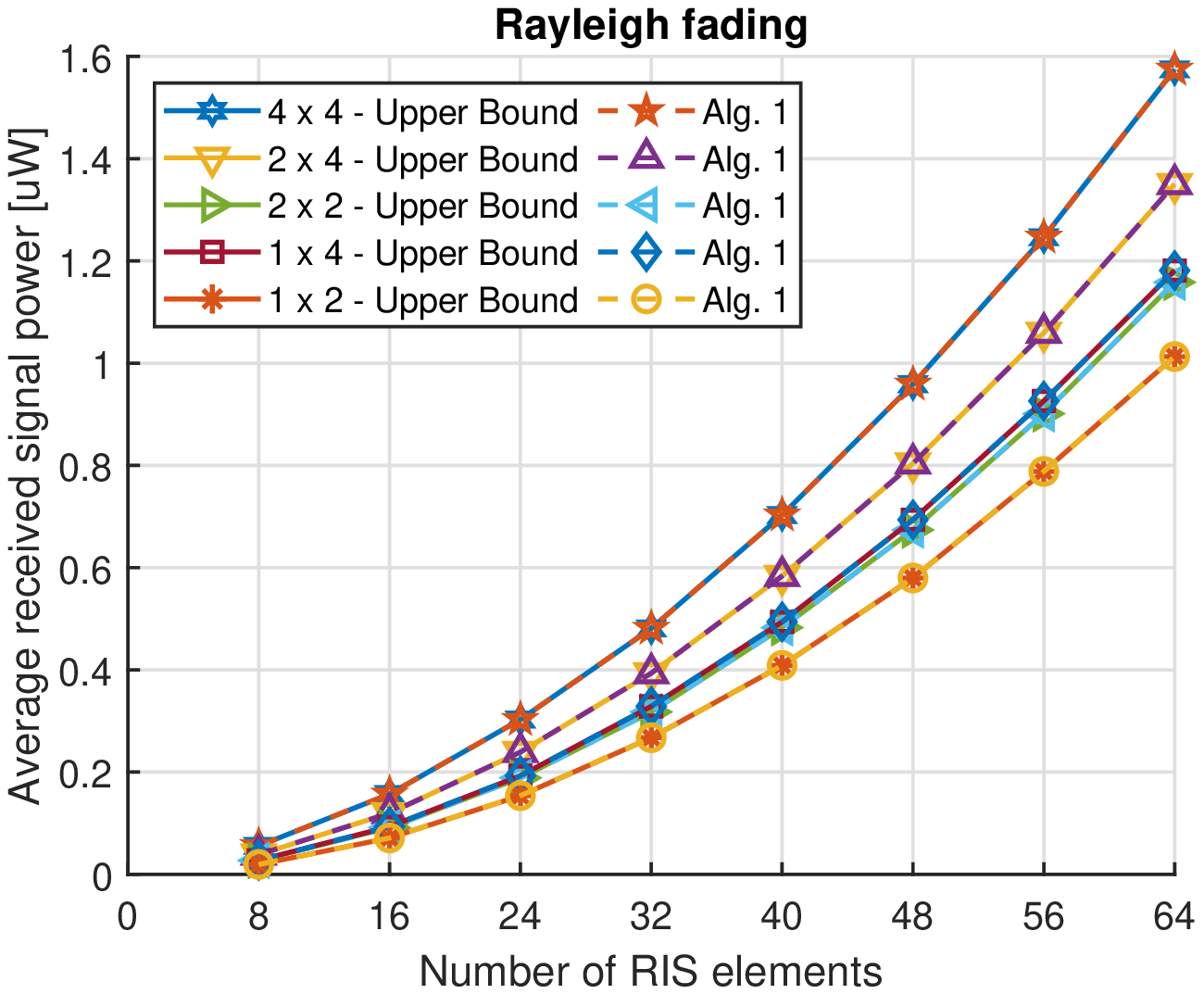}
\includegraphics[width=0.42\textwidth]{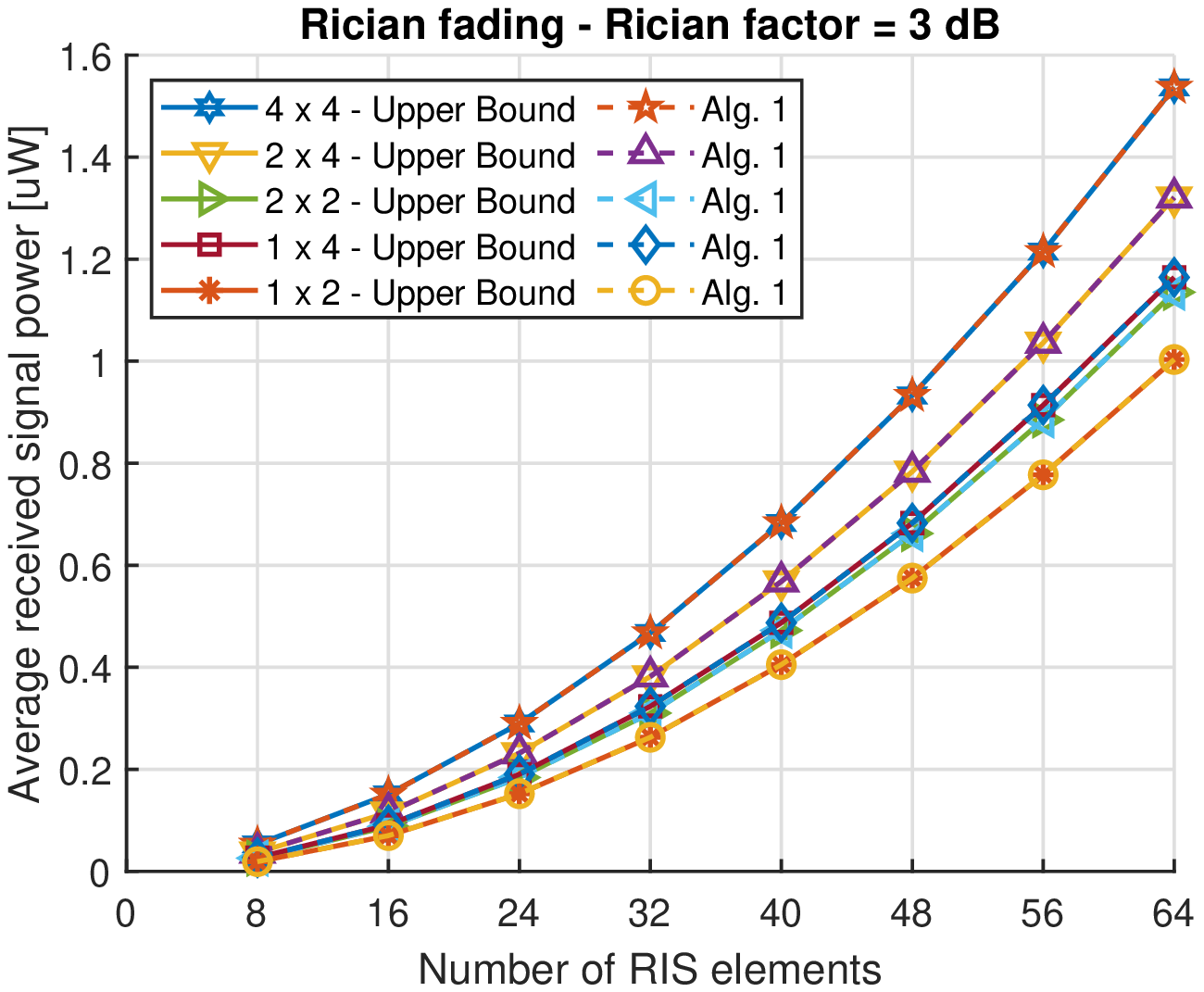}
\caption{Average received signal power in $N_R\times N_T$ systems without direct link aided by fully connected RISs working in reflective mode.}
\label{fig:pr-MIMO}
\end{figure*}
\begin{figure*}[t]
\centering
\includegraphics[width=0.42\textwidth]{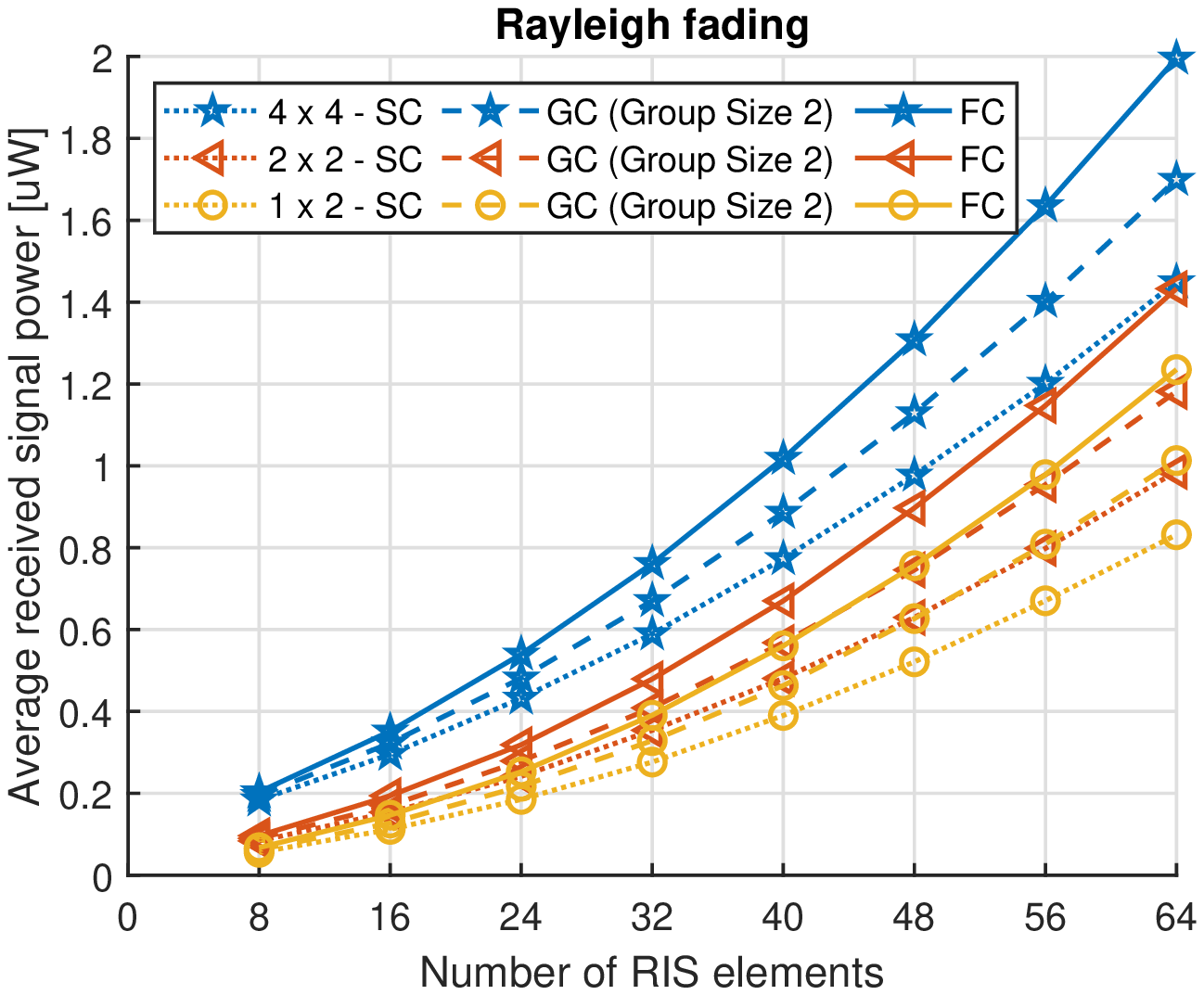}
\includegraphics[width=0.42\textwidth]{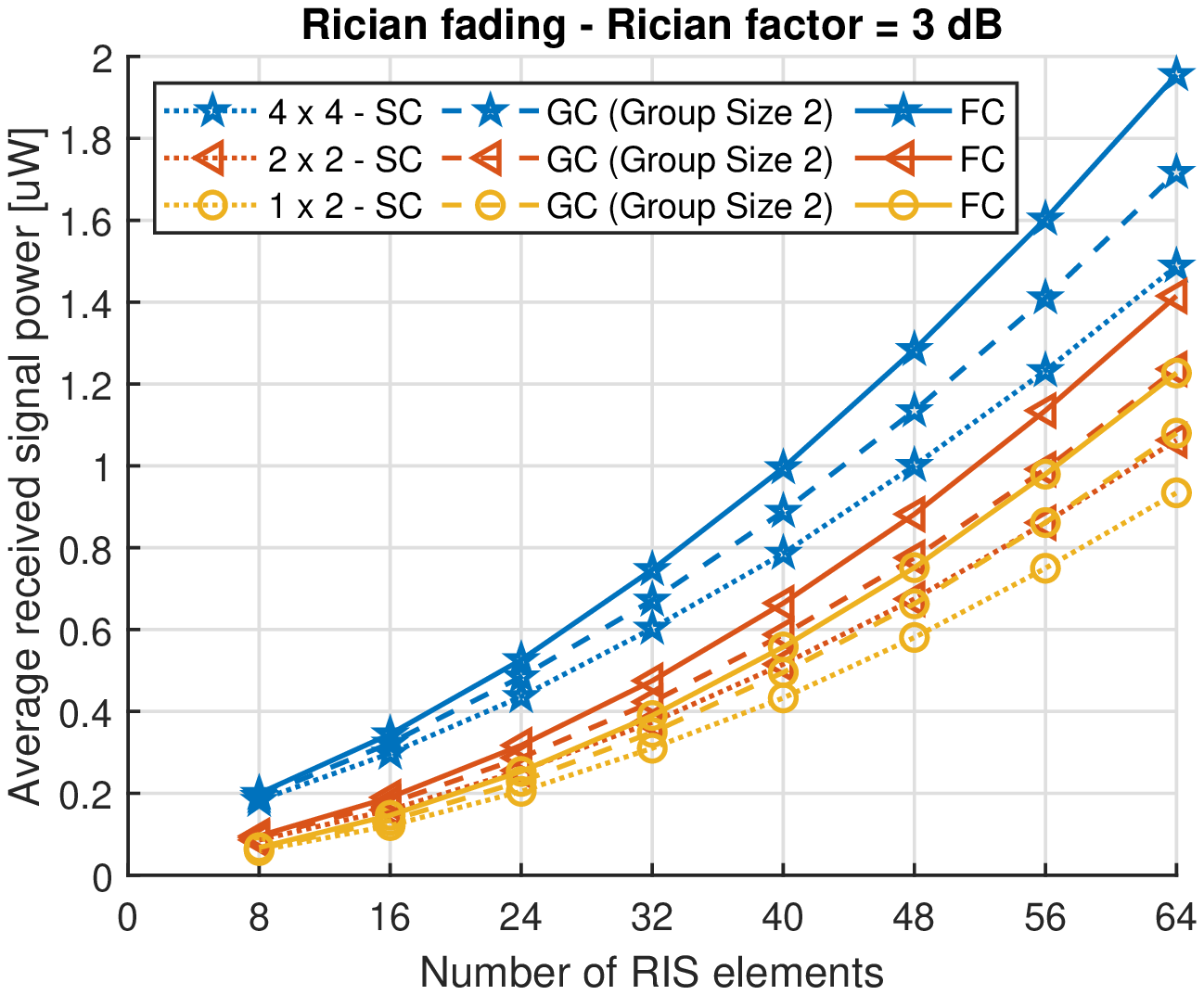}
\caption{Average received signal power in $N_R\times N_T$ systems aided by single connected ``SC'', group connected ``GC'', and fully connected ``FC'' BD-RISs working in reflective mode and optimized through Alg.~\ref{alg:theta-mimo}.}
\label{fig:pr-MIMO-HRT}
\end{figure*}

For general systems in which group connected RISs are considered or the direct links are not negligible, tight performance upper bounds are not available.
In this case, we notice that maximizing \eqref{eq:SR(H)} is equivalent to maximizing
\begin{equation}
S_R=P_{T}\left|\mathbf{z}_1\left(\mathbf{G}_{RT}+\mathbf{G}_{RI}\boldsymbol{\Theta}\mathbf{H}_{IT}\right)\mathbf{z}_2\right|^{2},\label{eq:SR(T)}
\end{equation}
where $\mathbf{z}_1\in\mathbb{C}^{1\times N_{R}}$ and $\mathbf{z}_2\in\mathbb{C}^{N_{T}\times1}$ are auxiliary variables such that $\left\|\mathbf{z}_1\right\|=1$ and $\left\|\mathbf{z}_2\right\|=1$.
Furthermore, maximizing \eqref{eq:SR(T)} is similar to the problem of maximizing the received signal power in \eqref{eq:PR-MIMO}, solved through Alg.~\ref{alg:theta-mimo} in Section~\ref{subsec:su-mimo-general}.
Thus, our sub-optimal strategy provided by Alg.~\ref{alg:theta-mimo} can be readily applied to solve also this maximization problem.
Multi-user \gls{miso} systems aided by a BD-RIS working in hybrid mode can be similarly optimized by directly applying the discussion made in Section~\ref{sec:design-transmissive}.

\section{Numerical Results}
\label{sec:results}

Let us consider a two-dimensional coordinate system, in which the $y$-axis represents the height above the ground in meters (m).
The transmitter and the receiver are located at $(0,0)$ and $(52,0)$, respectively.
The RIS is located at $(50,2)$ and is equipped with $N_I$ antennas.
Note that the simulation setting is similar to the setting adopted in \cite{she20}, with the RIS close to the receiver to maximize the gain brought by the RIS.
Nevertheless, the conclusions drawn in this study are not impacted by the position of the RIS.
The distance-dependent path loss is modeled as $L_{ij}(d_{ij})=L_{0}d_{ij}^{-\alpha_{ij}}$, where $L_{0}$ is the reference path loss at distance 1 m, $d_{ij}$ is the distance, and $\alpha_{ij}$ is the path loss exponent for $ij\in\{RT,RI,IT\}$.
We set $L_{0}=-30$ dB, $\alpha_{RT}=3.5$, $\alpha_{RI}=2.8$, $\alpha_{IT}=2$, and $P_{T}=10$ W.
For the small-scale fading, the channels are modeled with both Rayleigh and Rician fading, given by
\begin{equation}
\mathbf{h}_{ij}=\sqrt{L_{ij}}\left(\sqrt{\frac{K_F}{1+K_F}}\mathbf{h}_{ij}^{\mathrm{LoS}}+\sqrt{\frac{1}{1+K_F}}\mathbf{h}_{ij}^{\mathrm{NLoS}}\right),\label{eq:Rician Channel Model}
\end{equation}
where $K_F$ refers to the Rician factor, while $\mathbf{h}_{ij}^{\mathrm{LoS}}$ and $\mathbf{h}_{ij}^{\mathrm{NLoS}}\sim\mathcal{CN}\left(\boldsymbol{0},\mathbf{I}\right)$ represent the small-scale \gls{los} and \gls{nlos} (Rayleigh fading) components, respectively, for $ij\in\{RT,RI,IT\}$.
To model Rician fading channels, we consider $K_F=3$ dB.

\subsection{RIS-Aided Single-User SISO Systems: Reflective Mode}

We start by analyzing the performance of \gls{siso} systems aided by single, group, and fully connected RISs working in reflective mode.
In Fig.~\ref{fig:pr}, we report the average received signal power given in \eqref{eq:PR} obtained by optimizing the scattering matrix $\boldsymbol{\Theta}$ through our optimal strategy proposed in Section~\ref{subsec:group-connected}, for different group sizes.
We compare these results with the average received signal power upper bound given by
\begin{equation}
\bar{P}_{R}^{\mathrm{Group}}=P_{T}\left(\left\vert h_{RT}\right\vert+\sum_{g=1}^{G}\left\Vert \mathbf{h}_{RI,g}\right\Vert \left\Vert \mathbf{h}_{IT,g}\right\Vert \right)^{2},\label{eq:PR-UB}
\end{equation}
as derived in \cite{she20}.
The upper bound \eqref{eq:PR-UB}, valid for group connected RISs, boils down to
\begin{equation}
\bar{P}_{R}^{\mathrm{Single}}=P_{T}\left(\left\vert h_{RT}\right\vert+\sum_{n_{I}=1}^{N_{I}}\left|\left[\mathbf{h}_{RI}\right]_{n_{I}}\left[\mathbf{h}_{IT}\right]_{n_{I}}\right|\right)^2\label{eq:PR-UB-SC}
\end{equation}
for single connected RISs and to
\begin{equation}
\bar{P}_{R}^{\mathrm{Fully}}=P_{T}\left(\left\vert h_{RT}\right\vert+\left\Vert\mathbf{h}_{RI}\right\Vert\left\Vert \mathbf{h}_{IT}\right\Vert\right)^2\label{eq:PR-UB-FC}
\end{equation}
for fully connected RISs\footnote{The upper bounds \eqref{eq:PR-UB-SC} and \eqref{eq:PR-UB-FC} scale as $\mathcal{O}(\Gamma(1.5)^4N_I^2)$ and $\mathcal{O}(N_I^2)$, respectively, where $\Gamma(\cdot)$ refers to the gamma function, in the presence of Rayleigh fading channels \cite{she20}.
Thus, fully connected RISs provide a power gain of $1/\Gamma(1.5)^{4}=16/\pi^2$ over the single connected RISs when $N_I\to\infty$.
}.
As expected, we observe that the upper bounds are exactly achieved by our closed-form solution.
The upper bounds \eqref{eq:PR-UB}, \eqref{eq:PR-UB-SC}, and \eqref{eq:PR-UB-FC} can be numerically achieved also by the quasi-Newton method used in \cite{she20}.
However, the convergence of the quasi-Newton method to a global optimum cannot be mathematically proved.
Fully connected RISs achieve the same performance in both Rayleigh and Rician fading conditions.
Besides, single and group connected RISs benefit from the \gls{los} component, achieving higher power with Rician fading, agreeing with \cite{she20}\footnote{Considering correlated Rayleigh fading distributed according to the exponential correlation model with coefficient $\rho=0.5$, the performance of single and fully connected RISs remains unchanged compared to i.i.d. Rayleigh fading.
Besides, the performance of group connected RISs slightly reduces.}.
Remarkably, the \gls{siso} system without RIS can only achieve an average received signal power of $9.88$ nW.

\begin{figure*}[t]
\centering
\includegraphics[width=0.42\textwidth]{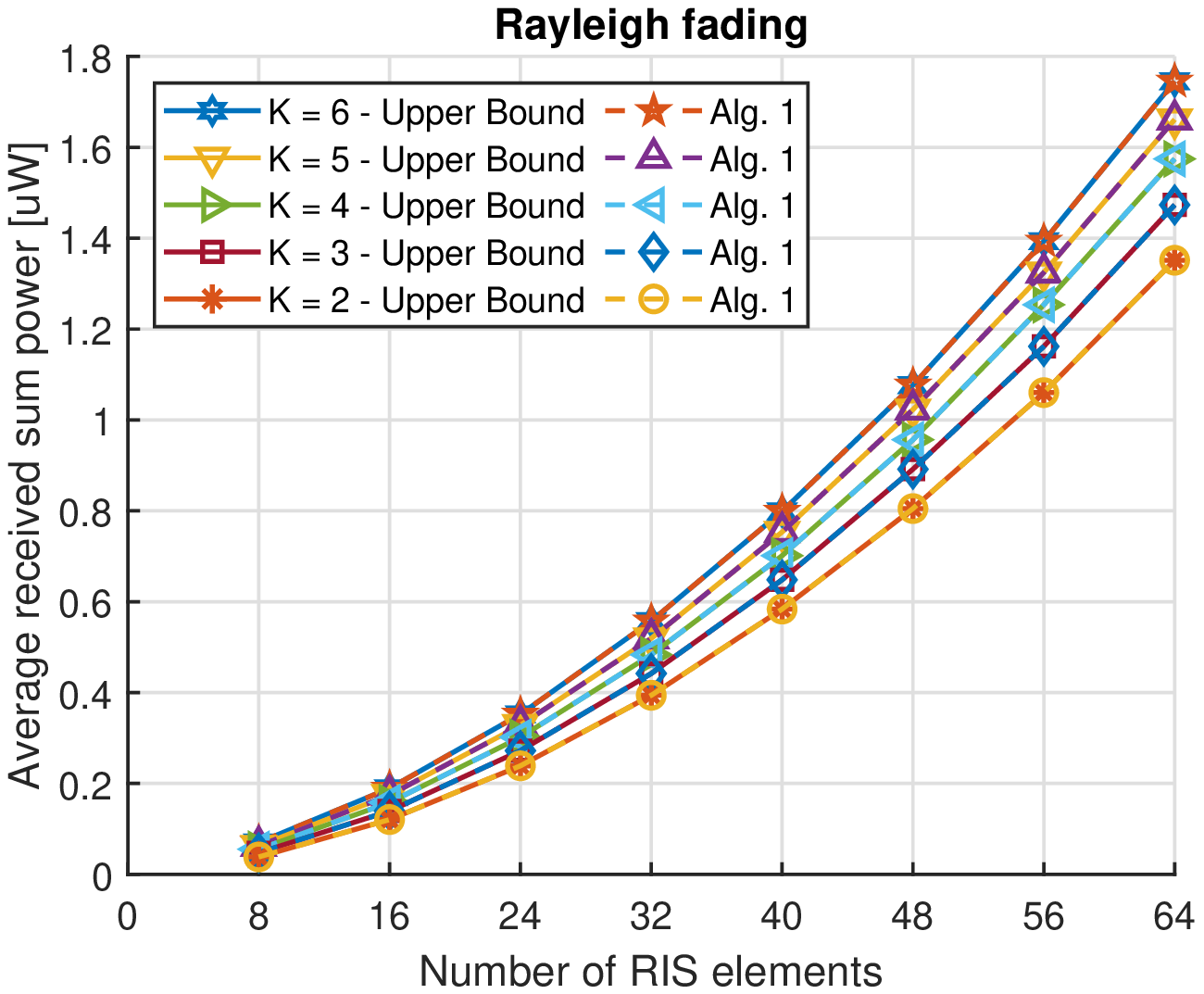}
\includegraphics[width=0.42\textwidth]{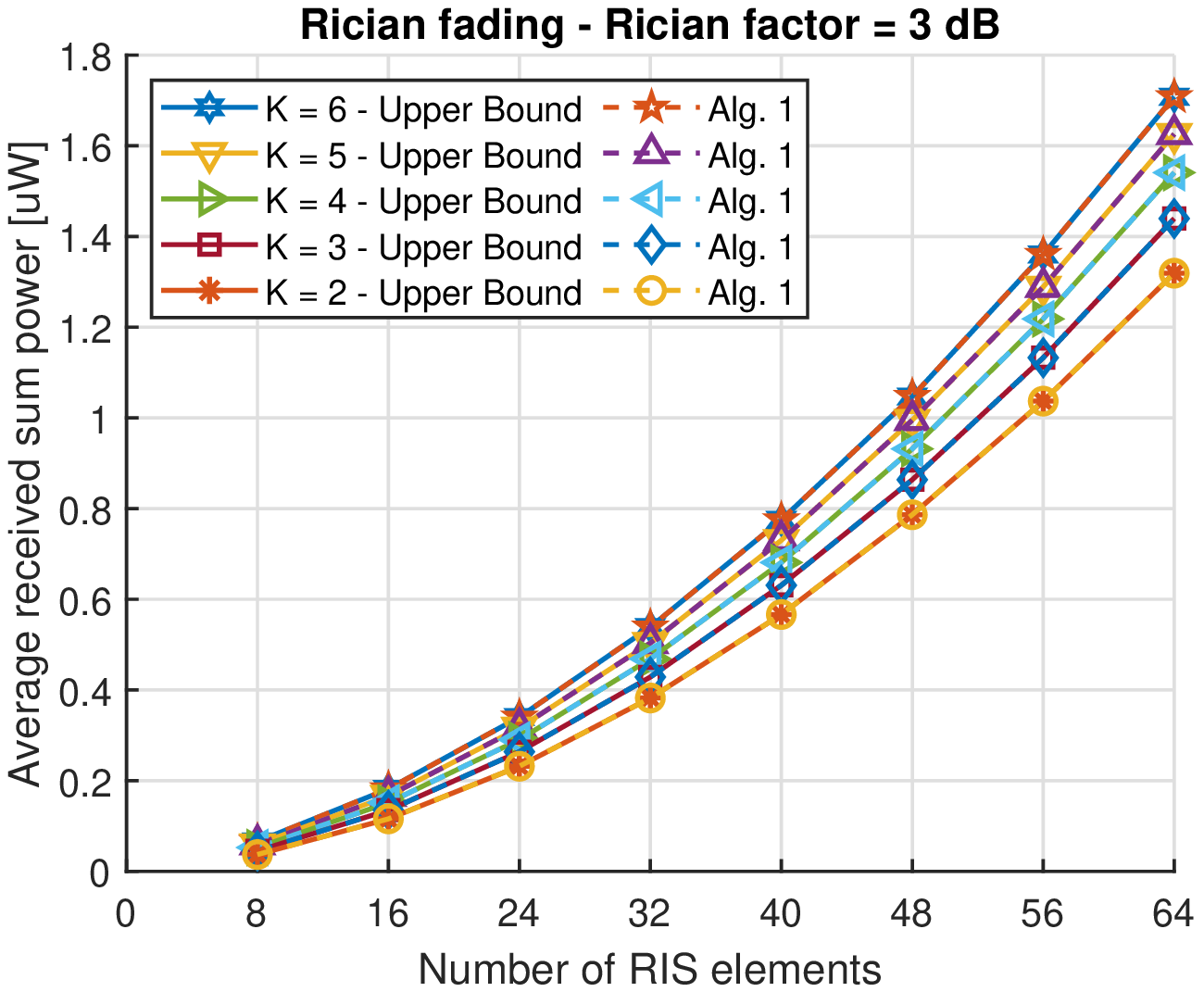}
\caption{Average received sum power in multi-user MISO systems with $N_T=4$ without direct link aided by fully connected RISs working in reflective mode.}
\label{fig:pr-MU-MIMO}
\end{figure*}
\begin{figure*}[t]
\centering
\includegraphics[width=0.42\textwidth]{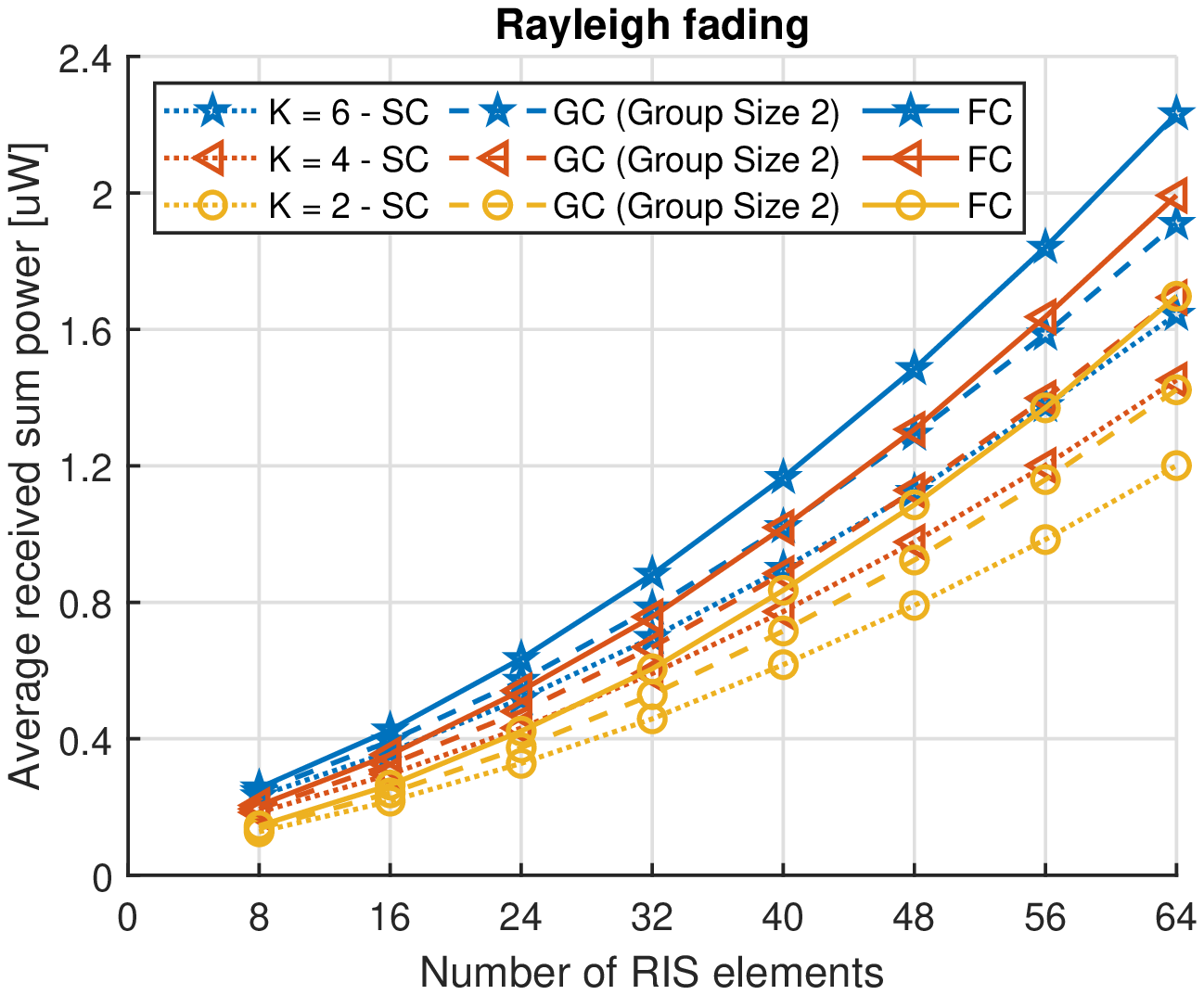}
\includegraphics[width=0.42\textwidth]{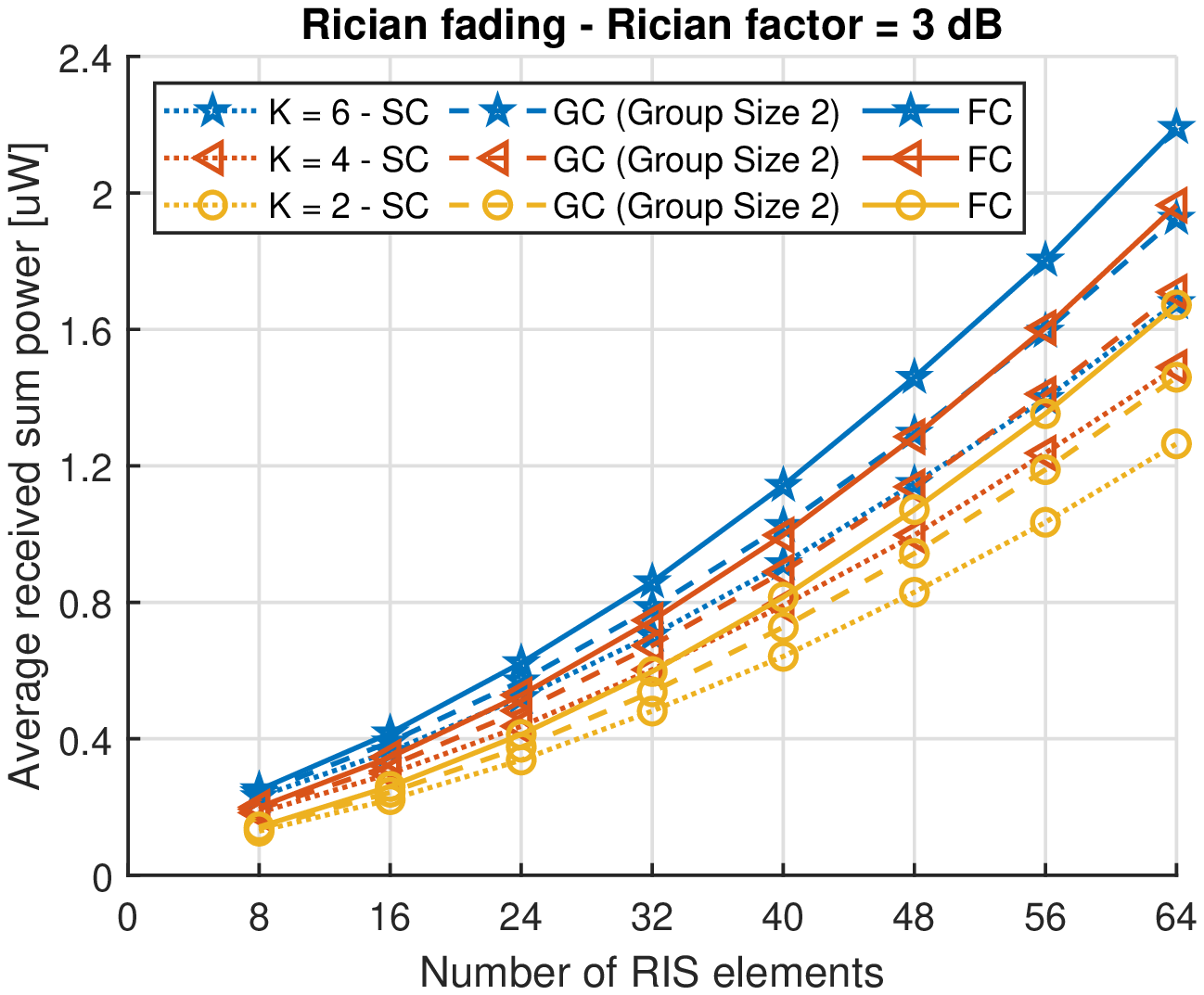}
\caption{Average received sum power in multi-user MISO systems with $N_T=4$ aided by single connected ``SC'', group connected ``GC'', and fully connected ``FC'' BD-RISs working in reflective mode and optimized through Alg.~\ref{alg:theta-mimo}.}
\label{fig:pr-MU-MIMO-HRT}
\end{figure*}

\subsection{RIS-Aided Single-User SISO Systems: Transmissive Mode}

We now consider \gls{siso} systems aided by a BD-RIS working in transmissive mode.
The receiver is now located in $(52,4)$ and we set $\alpha_{RT}=4$ to model a weak direct link.
As previously discussed, we assume the transmitter to be in sector 1 and the receiver in sector 2.
Thus, the odd entries of the channel $\mathbf{h}_{RI}$ and the even entries of the channel $\mathbf{h}_{IT}$ are forced to zero.
In Fig.~\ref{fig:pr-transmissive-mode}, we report the received signal power achieved by the optimal design strategy proposed in Section~\ref{subsec:group-connected} and its upper bounds.
We observe that the received signal power upper bounds are exactly achieved by our solution involving Alg.~\ref{alg:theta-siso}.
Thus, our optimal design strategy can be successfully applied to BD-RIS working in both reflective and transmissive modes.

\subsection{RIS-Aided Single-User MIMO Systems}

We analyze the performance of RIS-aided single-user \gls{mimo} and \gls{miso} systems, in which the RIS works in reflective mode.
In Fig.~\ref{fig:pr-MIMO}, we report the received signal power \eqref{eq:PR-MIMO} for $N_R\times N_T$ systems aided by a fully connected RIS, and with negligible direct link.
The received signal power obtained by the exact solution provided by Alg.~\ref{alg:theta-siso} is compared with its upper bound \eqref{eq:PR-MIMO-UB}.
We observe that the performance upper bounds are exactly achieved by our solution.
Furthermore, higher performance is obtained by increasing the number of antennas $N_R$ and $N_T$.
Rayleigh fading channels allow reaching a slightly higher received signal power since they offer richer scattering.
%
In Fig.~\ref{fig:pr-MIMO-HRT}, we consider single-user \gls{mimo} and \gls{miso} systems aided by a single, group, or fully connected RIS.
Alg.~\ref{alg:theta-mimo} is used to maximize the received signal power \eqref{eq:PR-MIMO} by optimizing the RIS scattering matrix.
As expected, fully connected RISs achieve higher received signal power than group connected RISs, which in turn outperform single connected RISs.
The performance gap between fully connected and single connected architectures is slightly higher in Rayleigh fading conditions.

\subsection{RIS-Aided Multi-User MISO Systems}

We now consider the weighted sum power maximization problem in RIS-aided multi-user \gls{miso} systems, with the RIS working in reflective mode.
In our simulations, all $K$ receivers are placed in $(52,0)$, and we set $N_T=4$ and $\alpha_k=1$ $\forall k$.
In Fig.~\ref{fig:pr-MU-MIMO}, we consider multi-user \gls{miso} systems aided by a fully connected RIS, and with negligible direct links between transmitter and receivers.
The received sum power \eqref{eq:SR} is maximized by applying the optimal precoding $\mathbf{w}=\mathbf{v}_{\text{max}}\left(\mathbf{S}\right)$, and by designing $\boldsymbol{\Theta}$ through Alg.~\ref{alg:theta-siso}.
We compare this received sum power with its upper bound \eqref{eq:SR-UB}.
As expected, the solution offered by Alg.~\ref{alg:theta-siso} is optimal as it exactly achieves the performance upper bounds.
%
In Fig.~\ref{fig:pr-MU-MIMO-HRT}, we report the weighted sum power of multi-user \gls{miso} systems aided by a single, group, or fully connected RIS, as given in \eqref{eq:SR}.
In these systems, Alg.~\ref{alg:theta-mimo} is used to optimize the RIS scattering matrix.
Fully connected RISs achieve the highest performance over group and single connected BD-RISs, while single connected RISs obtain the lowest weighted sum power.
Besides, the weighted sum power increases with the number of receivers $K$ because of the higher diversity offered.

\subsection{Computational Complexity}

Finally, we assess the computational complexity of our optimal design strategy.
In fully connected architectures, the complexity growth of Alg.~\ref{alg:theta-siso} as a function of $N_I$ is given by the complexity of eigenvalue decomposition, that is $\mathcal{O}(N_I^3)$.
This is less than the complexity of the quasi-Newton optimization adopted in previous literature, which is $\mathcal{O}(N_I^2(N_I+1)^2/4)$ for each iteration \cite{she20}.
In group connected architectures with group size $N_G$, the block diagonal scattering matrix is designed by running $G=N_I/N_G$ times Alg.~\ref{alg:theta-siso}, with complexity $\mathcal{O}(N_G^3)$.
Thus, the complexity of our solution is $\mathcal{O}(N_G^2N_I)$ in this case, less than the complexity of quasi-Newton optimization, given by $\mathcal{O}(N_I^2(N_G+1)^2/4)$ for each iteration \cite{she20}.
In Fig.~\ref{fig:complexity}, the computational complexity of Alg.~\ref{alg:theta-siso} is compared with the complexity of the quasi-Newton method used in \cite{she20}.

\begin{figure}[t]
\centering
\includegraphics[width=0.42\textwidth]{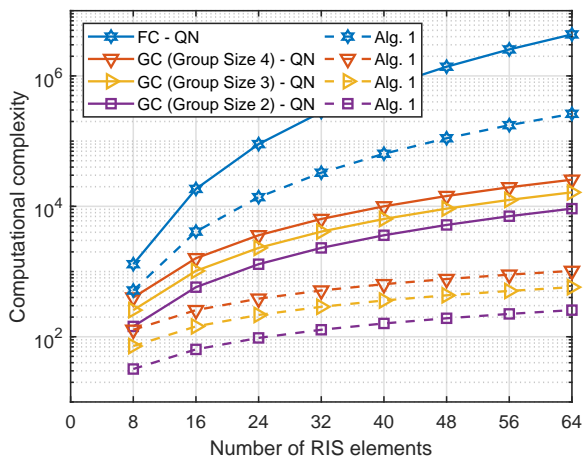}
\caption{Computational complexity versus the number of RIS elements. Alg.~\ref{alg:theta-siso} is compared with the quasi-Newton method ``QN'' for group connected ``GC'' and fully connected ``FC'' BD-RISs.}
\label{fig:complexity}
\end{figure}

\section{Conclusion}
\label{sec:conclusion}

We provide a low-complexity closed-form solution to design the global optimal scattering matrix in the case of group and fully connected RISs.
The resulting scattering matrix is proved to achieve exactly the received signal power upper bounds derived in \cite{she20}.
Our solution is upper bound-achieving for any channel realization since we do not pose assumptions on the channel distribution.
We first present a closed-form global optimal design strategy for RIS-aided single-user \gls{siso} systems.
Subsequently, our strategy is extended to single-user \gls{mimo} and multi-user \gls{miso} systems.
For systems aided by fully connected RISs and with negligible direct links, we provide tight performance upper bounds.
We show that such upper bounds can be exactly achieved with our optimal strategy.
Finally, we show that our algorithm is less complex than the iterative optimization methods applied to design BD-RISs in recent literature.
The complexity of our algorithm grows linearly (resp. cubically) with the number of RIS elements in the case of group (resp. fully) connected RISs.

Our optimal design strategy is expected to play a role in the solution of two problems related to the design of group and fully connected RISs.
Firstly, a possible research direction is to consider our strategy to further improve the design of discrete-value group and fully connected RISs.
Secondly, the optimal scattering matrix properties highlighted by our strategy could be exploited to enable efficient channel estimation for group and fully connected RISs.

\section*{Appendix}

\subsection{Proof of Proposition~\ref{pro:NI23}}
\label{subsec:proof1}

The matrix $\mathbf{A}$ is a linear combination of four outer products matrices, i.e.,
\begin{equation}
\mathbf{A}=\frac{1}{2}\mathbf{R}_{RI}+\frac{1}{2}\mathbf{R}_{RI}^T-\frac{1}{2}\mathbf{R}_{IT}-\frac{1}{2}\mathbf{R}_{IT}^T=\mathbf{B}\mathbf{B}^H,\label{eq:lin-comb}
\end{equation}
where the matrix $\mathbf{B}\in\mathbb{C}^{N_I\times 4}$ is introduced as $\mathbf{B}=1/\sqrt{2}[\hat{\mathbf{h}}_{RI}^H,\hat{\mathbf{h}}_{RI}^T,j\hat{\mathbf{h}}_{IT},j\hat{\mathbf{h}}_{IT}^*]$.
Since $\mathbf{B}$ is full rank, we have $r\left(\mathbf{A}\right)=r\left(\mathbf{B}\right)=\min\{4,N_I\}$.
Note that we assumed $\mathbf{h}_{RI}$ and $\mathbf{h}_{IT}$ to be independent in this discussion since the linearly dependent case has been trivially addressed.
Thus, it holds $r\left(\mathbf{A}\right)=N_I$ if $N_I\in\left\{2,3\right\}$.
The trace of $\mathbf{A}$ can be readily computed from \eqref{eq:lin-comb} by observing that $\text{Tr}\left(\mathbf{R}_{RI}\right)=1$ and $\text{Tr}\left(\mathbf{R}_{IT}\right)=1$ since $\Vert\hat{\mathbf{h}}_{RI}\Vert=1$ and $\Vert\hat{\mathbf{h}}_{IT}\Vert=1$.
By applying the trace linearity property, we have $\text{Tr}\left(\mathbf{A}\right)=0$.

\subsection{Proof of Proposition~\ref{pro:NI4}}
\label{subsec:proof2}

To prove that $r\left(\mathbf{A}\right)=4$ and $\text{Tr}\left(\mathbf{A}\right)=0$, we can directly apply the Proof of Proposition~\ref{pro:NI23}.
To prove that $\mathbf{A}$ has two positive and two negative eigenvalues, we use the fact that $\mathbf{A}=\mathbf{A}_{RI}-\mathbf{A}_{IT}$ is given by the sum of two symmetric matrices.
This proof is carried out by treating differently the cases $N_I=4$ and $N_I>4$.

In the case $N_I=4$, $\mathbf{A}$ is a full rank matrix, i.e., $\text{det}\left(\mathbf{A}\right)\neq0$.
We denote the decreasingly ordered eigenvalues of $\mathbf{A}_{RI}$ as $\delta_{RI,1},\ldots,\delta_{RI,N_I}$ and the decreasingly ordered eigenvalues of $\mathbf{A}_{IT}$ as $\delta_{IT,1},\ldots,\delta_{IT,N_I}$.
According to \cite{fie71}, $\text{det}\left(\mathbf{A}\right)$ can be lower bounded by
\begin{equation}
\min_P\prod_{n_I=1}^{N_I}\left(\delta_{RI,n_I}-\delta_{IT,Pn_I}\right)\leq\text{det}\left(\mathbf{A}\right),
\end{equation}
where the minimum is taken over all permutations of indices $1,\ldots,N_I$.
Since $\mathbf{A}_{RI}$ and $\mathbf{A}_{IT}$ are rank-2, it is always possible to find a permutation $P$ such that $\min_P\prod_{n_I=1}^{N_I}\left(\delta_{RI,n_I}-\delta_{IT,Pn_I}\right)=0$.
Thus, recalling that $\text{det}\left(\mathbf{A}\right)\neq0$, we obtain $\text{det}\left(\mathbf{A}\right)>0$.
Since $\mathbf{A}$ has four non-zero eigenvalues and $\text{Tr}\left(\mathbf{A}\right)=0$, $\text{det}\left(\mathbf{A}\right)>0$ implies the presence of two positive and two negative eigenvalues. 
This concludes the proof for $N_I=4$.

In the case $N_I>4$, we begin by noticing that $\mathbf{A}_{RI}$ and $\mathbf{A}_{IT}$ have two non-zero eigenvalues, both positives.
The reason is that the matrices are rank-2 by construction and positive semi-definite since both are the sum of two positive semi-definite matrices.
This means that $\delta_{RI,n_I},\delta_{IT,n_I}=0$ if $n_I>2$.
Furthermore, since $\mathbf{A}_{RI}$ and $\mathbf{A}_{IT}$ are Hermitian, we can apply Weyl's inequalities to study the eigenvalues of $\mathbf{A}$.
According to Weyl's inequality, we have
\begin{equation}
\delta_{i}\leq\delta_{RI,i-j}-\delta_{IT,N_I-j},\label{eq:weyl}
\end{equation}
valid for $i=1,\ldots,N_I$ and $j=0,\ldots,i-1$ \cite{hor12}.
Considering \eqref{eq:weyl} with $i=3$ and $j=0$, we have 
\begin{equation}
\delta_{3}\leq\delta_{RI,3}-\delta_{IT,N_I}=0,
\end{equation}
since $\delta_{RI,3}=0$ and $\delta_{IT,N_I}=0$.
Additionally, the dual Weyl's inequality gives
\begin{equation}
\delta_{RI,i+k-1}-\delta_{IT,k}\leq\delta_{i},\label{eq:weyl-dual}
\end{equation}
valid for $i=1,\ldots,N_I$ and $k=1,\ldots,N_I-i+1$ \cite{hor12}.
Considering \eqref{eq:weyl-dual} with $i=3$ and $k=3$, we have
\begin{equation}
0=\delta_{RI,5}+\delta_{IT,3}\leq\delta_{3},
\end{equation}
since $\delta_{RI,5}=0$ and $\delta_{IT,3}=0$, yielding $\delta_{3}=0$.
Now, we consider \eqref{eq:weyl} with $i=N_I-2$ and $j=0$ to obtain
\begin{equation}
\delta_{N_I-2}\leq\delta_{RI,N_I-2}-\delta_{IT,N_I}=0,
\end{equation}
since $\delta_{RI,N_I-2}=0$ and $\delta_{IT,N_I}=0$.
Additionally, \eqref{eq:weyl-dual} with $i=N_I-2$ and $k=3$ gives
\begin{equation}
0=\delta_{RI,N_I}+\delta_{IT,3}\leq\delta_{N_I-2},
\end{equation}
since $\delta_{RI,N_I}=0$ and $\delta_{IT,3}=0$, yielding $\delta_{N_I-2}=0$.
Since $\delta_{3}=0$ and $\delta_{N_I-2}=0$, it has to be $\delta_{1},\delta_{2}>0$ and $\delta_{N_I-1},\delta_{N_I}<0$.
This concludes the proof for $N_I>4$.

\section*{Acknowledgment}
The authors would like to thank Matheus Manzatto de Castro of Imperial College London for the very fruitful discussions on this paper.

\bibliographystyle{IEEEtran}
\bibliography{IEEEabrv,main}

\end{document}